\documentclass[11pt, helvetica]{article}

\usepackage{booktabs} 
\usepackage[ruled]{algorithm2e} 

\renewcommand*{\algorithmcfname}{Commensurate Computation (CCOM)}

\usepackage{wrapfig,floatrow}
\usepackage{caption}

\usepackage{stfloats}

\usepackage{subcaption}
\usepackage{enumitem}
\usepackage{xspace}
\usepackage[toc,page]{appendix}

\usepackage{verbatim}
\usepackage{times}
\usepackage{amsmath, amsthm, amssymb}
\usepackage{color}
\usepackage{graphicx}
\usepackage{cite}
\usepackage{epsfig}
\usepackage{url}
\usepackage{xcolor}
\usepackage{xspace}
\usepackage{pgfgantt}
\usepackage{booktabs}
\usepackage{authblk}
\usepackage[letterpaper,top=1.0in, bottom=1.0in, left=1.0in, right=1.0in]{geometry}

\newif\ifcomments
\commentstrue
\renewcommand{\paragraph}[1]{\vspace{0.2em} \noindent \textbf{#1}}
 
\newcommand{\whp}{with high probability}

\newtheorem{theorem}{Theorem}

\newtheorem{lemma}{Lemma}

\newcommand{\POW}{PoW\xspace}
\newcommand{\AlgA}{\textsc{SybilControl\xspace}}
\newcommand{\AlgB}{\textsc{CCom}\xspace}
\newcommand{\AlgBRS}{\textsc{ECCom}\xspace}

\newcommand{\E}{\textsc{Elastico}\xspace}

\newcommand{\resp}{\textit{responsiveness}\xspace}

\newcommand{\Diffuse}{\textsc{Diffuse}\xspace}

\newcommand{\gnew}{g_{\mbox{\tiny new}}}

\newcommand{\ppuzzle}{purge\ puzzle\xspace}
\newcommand{\epuzzle}{entrance\ puzzle\xspace}

\newcommand{\tstamp}{\mathcal{T}}

\newcommand{\Old}{\mathcal{S}_{\mbox{\tiny old}}}
\newcommand{\Current}{\mathcal{S}}

\newcommand{\sOld}{\mathcal{S}'_{\mbox{\tiny old}}}
\newcommand{\sCurrent}{\mathcal{S'}}   

\newcommand{\qcomm}{committee\xspace}

\newcommand{\defn}[1]{\textbf{\emph{#1}}}

\newcommand{\cgoal}{Committee Goal\xspace}
\newcommand{\sgoal}{Population Goal\xspace}

\begin{document}
\title{Proof of Work Without All the Work:\\ \vspace{-5pt} {\Large  Computationally Efficient Attack-Resistant Systems}}

\date{}

\author[1]{Diksha Gupta}
\author[2]{Jared Saia \thanks{This work is supported by the National Science Foundation grants CNS-1318880 and CCF-1320994.}}
\author[3]{Maxwell Young\thanks{This work is supported by the National Science Foundation grant CCF 1613772 and by a research gift from C Spire.}}
\affil[1]{\small Dept. of Computer Science, University of New Mexico, NM, USA\hspace{6cm} \mbox{\texttt{dgupta@unm.edu}}}
\affil[2]{\small Dept. of Computer Science, University of New Mexico, NM, USA\hspace{6cm} \mbox{\texttt{saia@cs.unm.edu}}}
\affil[3]{\small Computer Science and Engineering Dept., Mississippi State University, MS, USA\hspace{6cm} \texttt{myoung@cse.msstate.edu}}


\maketitle 
\begin{abstract}

Proof-of-work (\POW) is an algorithmic tool used to secure networks by imposing a computational cost on participating devices. Unfortunately, traditional \POW schemes require that correct devices perform computational work perpetually, even when the system is not under attack.

We address this issue by designing a general \POW protocol that ensures two properties. First, the network stays secure.  In particular, the fraction of identities in the system that are controlled by an attacker is always less than $1/2$. Second, our protocol's computational cost is commensurate with the cost of an attacker.  In particular, the total computational cost of correct devices is a linear function of the attacker'€™s computational cost plus the number of correct devices that have joined the system.  Consequently, if the network is attacked, we ensure security with cost that grows linearly with the attacker's cost; and, in the absence of attack, our computational cost remains small.  We prove similar guarantees for bandwidth cost.

Our results hold in a dynamic, decentralized system where participants join and depart over time, and where the total computational power of the attacker is up to a constant fraction of the total computational power of correct devices. We demonstrate how to leverage our results to address important security problems in distributed computing including: Sybil attacks, Byzantine consensus, and Committee election.
\end{abstract}


\clearpage

\section{Introduction}\label{sec:intro}



Twenty-five years after its introduction by  Dwork and Naor~\cite{dwork:pricing}, \defn{proof-of-work (\POW)} is enjoying a research renaissance.  Succinctly, \POW is an economic tool to prevent abuse of a system by requiring participants to solve computational puzzles in order to access system resources or participate in group decision making.  In recent years, \POW~is playing a critical role in cryptocurrencies such as Bitcoin~\cite{nakamoto:bitcoin}, and other blockchain technologies~\cite{litecoin,blockstack,chain1,iota,ethereum,dashcoin,primecoin}. 

Yet, despite success with Bitcoin and its analogs, \POW~has not found widespread application in mitigating a range of malicious behaviors, such as Sybil attacks~\cite{douceur02sybil}, and application-layer distributed denial-of-service (DDoS) attacks. These are well-known and enduring security problems for which \POW~seems well-suited, and yet proposals~\cite{parno2007portcullis,wang:defending,kaiser:kapow,green:reconstructing,feng:design,waters:new,martinovic:wireless,borisov:computational,li:sybilcontrol} built around \POW~have  seen only limited deployment.  


\medskip

\noindent{\bf A Barrier to Widespread Use.} A major impediment to the widespread use of~\POW~is ``the work''. In particular, current \POW approaches have significant computational overhead, given that puzzles must always be solved, {\it even when the system is not under attack}.  This non-stop resource burning translates into a substantial energy --- and, ultimately, a monetary -- cost~\cite{economistBC,coindesk,arstechnica}.\footnote{For example, in 2015, the Economist calculated that Bitcoin consumes at least $1.46$ terawatt-hours of electricity per year, or enough to power $135,000$ American homes~\cite{economistBC}.}  Consequently,  \POW approaches are currently used primarily in applications where participants have a financial incentive to continually perform work, such as cryptocurrencies.  This is a barrier to wide-spread use of a technique that has the potential to fundamentally advance the field of cybersecurity.





In light of this, we seek to reduce the cost of~\POW~systems and focus on the following question: \textbf{Can we design~\POW~systems where the resource costs are low in the absence of attack, and grow commensurately with the effort expended by an attacker?}

In this paper, we design and analyze an algorithm that answers this question in the affirmative. Initially, our result is presented in the context of the Sybil attack; however,  it applies more generally to safeguarding a distributed system from any attack where an adversary seeks to obtain significantly more than its fair share of network resources. 


\subsection{Our Model}\label{sec:model-main}


Our system consists of virtual \defn{identifiers (IDs)}, and an attacker (or \defn{adversary}).  Each ID is either good or bad.   Each \defn{good} ID follows our algorithm, and all \defn{bad} IDs are controlled by the adversary. 
 
\medskip

\noindent{\bf Adversary.} Our assumption that a single adversary controls all the bad IDs pessimistically represents perfect collusion and coordination among the bad IDs. Bad IDs may deviate from our protocols in an arbitrary manner including sending incorrect or spurious messages. This type of adversary encapsulates the celebrated \defn{Sybil attack}~\cite{douceur02sybil}.  Section~\ref{sec:related-work} summarizes the extensive prior literature in this area.

We critically assume that the adversary controls up to an {\boldmath{$\alpha$}} \defn{fraction} of the computational resources in the network, for $\alpha$ a constant that is a parameter of our algorithms.  This assumption is standard in past \POW literature~\cite{nakamoto:bitcoin, andrychowicz2015pow, parno2007portcullis, walfish2010ddos,GiladHMVZ17, miller2014anonymous}.  

Our algorithms employ public key cryptography, and we make the usual cryptographic assumptions needed for this.  We assume that the adversary knows all of our algorithms, but does not know the private random bits of any good ID.

\medskip 
 
\noindent{\bf Communication.} All communication among good IDs occurs through a broadcast primitive, denoted by~{\bf \Diffuse}, which  allows  a  good  ID  to  send  a  value  to all other good IDs within a known and bounded amount of time, despite the presence of an adversary. We assume that when a message is diffused in the network, it is not possible to determine which ID initiated the diffusion of that message. Such a primitive is a standard assumption in \POW schemes~\cite{Garay2015,bitcoinwiki,GiladHMVZ17,Luu:2016}; see~\cite{miller:discovering} for empirical justification.  We assume that each message originating at a good ID is signed by the ID's private key. 


Time is discretized into \defn{rounds}. As a standard assumption, all IDs are assumed to be synchronized, but our algorithms can tolerate a small amount of skew. The length of a round is a positive constant times {\boldmath{$\Delta$}}, where $\Delta>0$ is some constant upper bound on the time to deliver a message. As the value of $\Delta$ is made larger, the communication latency becomes less of a factor; however, computational effort increases.


For simplicity, we initially assume that the time to diffuse a message is small in comparison to the time to solve computational puzzles.\footnote{A recent study of the Bitcoin network shows that the communication latency is 12 seconds in contrast to the 10 minutes block interval~\cite{croman2016scaling}, which motivates our assumption of computational latency dominating the communication latency in our system.} However, we later relax this assumption to account for a network latency which is upper bounded by $\Delta$.

We pessimistically assume that the adversary can send messages to any ID at will, and that it can read the messages diffused by good IDs before sending its own.



\medskip

\noindent{\bf Hashes and Puzzles}.  We assume that all IDs know a \defn{hash function} {\boldmath{$h$}} that maps bit strings to real 
    numbers in $[0,1)$.  We make the standard \defn{random 
    oracle assumption} about $h$~\cite{bellare1993random,canetti1997towards,koblitz2015random,nakamoto:bitcoin}.  
    This assumption is that when first computed on an input, 
    $x$, $h(x)$ is selected independently and uniformly at random from $[0,1)$, and that on subsequent computations of $h(x)$ the 
    same output value is always returned.\footnote{In practice, $h$ is a cryptographic hash function, such as SHA-2~\cite{sha2}, with 
    inputs and outputs of sufficiently large bit lengths.} 
A puzzle solution is typically a certain type of input to the hash function that achieves an output that is sufficiently small. 

Let {\boldmath{$\mu$}} denote the number of hash function queries that a good ID can make per unit time. Hence, we measure the computational power of a good ID in terms of $\mu$.

\medskip

\noindent{\bf Joins and Departures.} The system is dynamic with IDs joining and departing over time (i.e., \defn{churn}), subject to the constraint that at most a constant fraction of the good IDs can join or depart in any round. Maintaining performance guarantees amidst churn is often challenging in decentralized systems~\cite{Augustine:2013:SSD:2486159.2486170,7354403,Augustine2015,Augustine:2012,augustine:fast}.  We pessimistically assume that all join and departure events are scheduled in a worst-case fashion by the adversary. 

 We assume that all good IDs announce their departure to the network.  In practice, this assumption could be relaxed through the use of heartbeat messages that are periodically sent out to indicate that an ID is still in the network.


The minimum number of good IDs in the system at any point is assumed to be at least {\boldmath{$n_0$}}.  Our goal is to provide security and performance guarantees  for {\boldmath{$O(n_0^{\gamma})$} } joins and departures of IDs, for any desired constant {\boldmath{$\gamma$}} $ \geq 1$. In other words, the guarantees on our system hold with high probability (w.h.p.)\footnote{With probability at least $1-n_0^{-c}$ for any desired $c\geq 1$.} over this polynomial number of dynamic events. This implies security despite a system size that may vary wildly over time, increasing and decreasing polynomially in $n_0$ above a minimum number of $n_0$ good IDs. 

We impose a loose constraint on the rate of departures: at most an $\epsilon_0$-fraction of good IDs may join/depart in any single round, where {\boldmath{$\epsilon_0>0$}} is a small, known constant.   Note that this constraint still allows for an amount of dynamism that is linear in the current system size. 


\subsection{Our Goals.}\label{sec:our-problems}

We have two key security goals.\medskip

\noindent{\defn{\sgoal}{\bf:}} Ensure that the fraction of bad IDs in the entire system is always less than $1/2$.
\medskip

Achieving this goal will help ensure that system resources -- such as jobs executed on a server, or bandwidth obtained from a wireless access point -- consumed by the bad IDs are (roughly) proportional to the adversary's computational power.  Possible application domains include: content-sharing overlays~\cite{falkner:profiling,steiner:global}; and open cloud services~\cite{Mohaisen:2013:TDC:2484313.2484332,Anderson:2004:BSP:1032646.1033223,Chandra:2009:NUD:1855533.1855535,WeissmanSGRNC11}. 

\medskip

  
\medskip  

\noindent{\defn{\cgoal}{\bf:}} Ensure there always exists a committee that (i) is known to all good IDs; (ii) is of \emph{scalable} size i.e., of size $\Theta(\log n_0)$; and (iii) contains less than a $1/2$  fraction of bad IDs.


\medskip

Achieving the \cgoal will ensure that the committee can solve the \defn{Byzantine consensus}~problem~\cite{lamport1982byzantine}.  In this problem, each good ID has an initial input bit. The goal is for (i) all good IDs to decide on the same bit; and (ii) this bit to equal the input bit of at least one good ID.  

Byzantine consensus enables participants in a distributed network to reach agreement on a decision, even in the presence of a malicious minority.  Thus, it is a fundamental building block for many cryptocurrencies~\cite{BonneauMCNKF15,eyal2016bitcoin,cryptoeprint:2015:521,GiladHMVZ17}; trustworthy computing~\cite{castro1998practical, castro2002practical,SINTRA,cachin:secure,kotla2007zyzzyva,clement-making, 1529992}; peer-to-peer networks~\cite{oceanweb,adya:farsite}; and databases~\cite{GPD,preguica2008byzantium, zhao2007byzantine}.  Establishing Byzantine consensus via the use of committees is a common  approach; for examples, see~\cite{GiladHMVZ17,Luu:2016,KSSV}. 

\medskip

\noindent We note that it is possible to reduce the fraction $1/2$ in either goal by reducing the value of $\alpha$. Second, we expect the \sgoal to be particularly relevant for small or medium-sized networks, where it is manageable to maintain membership information.  In contrast, the \cgoal may be more relevant for larger networks, since it does not require maintenance of large membership lists.  


\subsection{Main Results}\label{sec:main}

We measure \defn{computational cost} as the effort required to solve computational puzzles (see Section~\ref{sec:puzzles} for details), and we measure \defn{bandwidth cost} as the number of calls to \Diffuse. Let $T_C$ and $T_B$ denote the total computational and total bandwidth costs, respectively, incurred by the adversary.  Let $\gnew$ denote the number of good IDs that have joined the system.  The \defn{lifetime} of the system is defined as the time until at least $ O(n_0^{\gamma})$ i.e., polynomially in $n_0$ join or leave events.

\begin{theorem}\label{thm:main1}
 Assume that the adversary holds at most an $\alpha = 1/6$-fraction of the total computational power of the network. Then, w.h.p. our algorithm (\AlgB in Section \ref{sec:server1}) ensures the following properties hold over the lifetime of the system: \vspace{4pt} 
\begin{itemize}[topsep=0pt,parsep=0pt,partopsep=0pt,leftmargin=13pt,labelwidth=6pt,labelsep=4pt]
\item \sgoal: The fraction of bad IDs in the system is always less than $1/2$.
\item \cgoal: There is a committee always known to all IDs such that at all times the committee  contains (1) a logarithmic number of IDs; and (2) less than a $1/2$  fraction of bad IDs.
\item The cumulative computational cost to the good IDs is $O\left(T_C+ \gnew \right)$. \vspace{2pt}
\item The cumulative bandwidth cost to the good IDs is $\tilde{O}\left(T_B + \gnew \right)$. 
\end{itemize}
\end{theorem}


Note that the computational and bandwidth costs incurred by the good IDs grow slowly with the cost incurred by the adversary. Recalling our discussion at the beginning of this section, this is precisely the type of result we sought. When there is no attack on the system, the costs are low and solely a function of the number of good IDs; there is no excessive overhead. But as the adversary spends more to attack the system, the costs required to keep the system secure grow commensurately with the adversary's costs. \medskip

\paragraph{Empirical Results and Applications.}  In Section~\ref{sec:empirical}, we empirically evaluate our algorithm. For a variety of networks, our algorithm significantly reduces computational cost.  

In Section~\ref{sec:applications}, we discuss applications of our algorithm to (1) Byzantine consensus; and (2) Elastico, which is a committee election algorithm that can be used to achieve consensus on Blockchains for cryptocurrencies.


\subsection{Related Work} \label{sec:related-work}


A preliminary version of our results appeared as an extended abstract in~\cite{pow-without}. This current work makes several significant extensions to these earlier results. First, we provide a method to greatly reduce the amount of state that needs to be stored by good IDs; see Section~\ref{sec:reducing-state}. Second, we demonstrate that our result is robust to a bounded amount of communication delay; see Section~\ref{sec:bounded-latency}. Third, we extend the empirical evaluation of our algorithm under aggressive attack; see Section~\ref{sec:empirical}.  Additionally, this work contains the full proofs, along with an expanded discussion of related work and future open problems.

\medskip

\noindent{\bf The Sybil Attack.} Our work applies -- but is not limited -- to the Sybil attack~\cite{douceur02sybil}. There is large body of literature on mitigating such attacks (for example, see surveys~\cite{newsome:sybil,mohaisen:sybil,john:soft,Dinger:2006}), and additional work documenting real-world Sybil attacks~\cite{bitcoin-sybil,6503215,Yang:2011:USN:2068816.2068841,Tran:2009:SOC:1558977.1558979}.  {\it Critically, none of these prior results have the desired property that the resource costs to the good IDs grow commensurately with the cost incurred by the adversary.}  We note that this can be characterized as a \defn{resource-competitive} approach~\cite{gilbert:making, gilbert:near,king:conflict,bender:how,ICALP15, daniICJournal17, aggarwal2016secure,gilbert:resource,Bender:2015:RA:2818936.2818949,zamani2017torbricks}. However, for simplicity, we omit a discussion of resource-competitive algorithms since it is not critical to understanding our results.

We note that obtaining large numbers of machines is expensive; computing power costs money, whether obtained via Amazon AWS~\cite{amazon-aws} or a botnet rental~\cite{anderson2013measuring}. This premise is borne out by Bitcoin and other cryptocurrencies where an adversary has a strong incentive to wield as much computational power as possible.


Beyond \POW, several other approaches have been proposed. In a wireless setting with multiple communication channels, Sybil attacks can be mitigated via \emph{radio-resource testing} which relies on the inability of the adversary to listen to many channels simultaneously~\cite{monica:radio,newsome:sybil,gilbert:sybilcast,gilbert:who}.  However, this approach may fail if the adversary can monitor most or all of the channels. Furthermore, the same shortcoming exists: the system must constantly perform tests to guarantee security which leads to wasted bandwidth in the absence of attack.

There are several results that leverage social networks to yield Sybil resistance~\cite{yu:survey,lesniewski-laas:whanau,yu:sybilguard,mohaisen:improving,wei:sybildefender,Yu:2010:Sybillimit}. However, social-network information may not be available in many settings. Another idea is to use network measurements to verify the uniqueness of IDs~\cite{bazzi:distinct,sherr:veracity,1550961,liu:mason,Gil-RSS-15,demirbas:rssi}, but these techniques rely on accurate measurements of latency, signal strength, or round-trip times, for example, and this may not always be possible.  Containment strategies are examined in overlays~\cite{danezis:sybil,scheideler:shell}, but the extent to which good participants are protected from the malicious actions of Sybil IDs is limited.

\medskip

\noindent{\bf \POW and Alternatives.} As a choice for~\POW, computational puzzles provide certain advantages. First, verifying a solution is much easier than solving the puzzle itself. This places the burden of proof on devices who wish to participate in a protocol rather than on a verifier.  In contrast, bandwidth-oriented schemes, such as~\cite{walfish2010ddos}, require verification that a sufficient number of packets have been received before any service is provided to an ID; this requires effort by the verifier that is proportional to the number of packets. 



A recent alternative to \POW is \defn{proof-of-stake (PoS)} where security relies on the adversary holding a minority stake in an abstract finite resource~\cite{abraham:blockchain}. When making a group decision, PoS weights each participant's vote using its share of the resource; for example, the amount of cryptocurrency held by the participant.  A well-known example is ALGORAND~\cite{GiladHMVZ17}, which employs PoS to form a committee.   A hybrid approach using both \POW and PoS has been proposed in the Ethereum system~\cite{ethereum-pos}.  We note that PoS can only be used in systems where the ``stake" of each participant is globally known.  Thus, PoS is typically used only in cryptocurrency applications. 

Recent work by Pass et. al~\cite{pass2016hybrid} proposes a consensus protocol for a dynamic cryptocurrency system, where they use the blockchain protocol to elect a new committee everyday. The committee runs a Byzantine Fault Tolerant (BFT) protocol on the transactions to generate a daily log, which is logged in the blockchain. They introduce a new performance measure called \resp~that is applicable to a cryptocurrency system only if it possesses the ability to log a transaction input to an honest ID within some function of the actual delay rather than that of a loose upper bound on the delay in the system. 

\medskip

\noindent Finally, when the number of bad participants is {\it assumed} to be always bounded, a number of results on adversarial fault-tolerant systems exist~\cite{castro:byzantine,cachin:secure,adya:farsite,cowling:hq,kubiatowicz:oceanstore,rodrigues:rosebud,rodrigues:design,rodrigues:large,halo:kapadia,salsa:nambiar,bortnikov:brahms,johansen:fireflies}.

\section{Our Algorithm}\label{sec:server1} 


In this section, we describe the main ingredients that go into our algorithm. 
 

\subsection{Computational Puzzles}\label{sec:puzzles}

All IDs have access to a hash function, {\boldmath{$h$}}, about which we make the standard \emph{random oracle assumption}~\cite{bellare1993random,canetti1997towards,koblitz2015random}.  Succinctly, this assumption is that when first computed on an input, $x$, $h(x)$ is selected independently and uniformly at random from the output domain, and that on subsequent computations of $h(x)$ the same output value is always returned.  We assume that both the input and output domains are the real numbers between $0$ and $1$.  In practice, $h$ may be a cryptographic hash function, such as SHA-2~\cite{sha2}, with inputs and outputs of sufficiently large bit lengths.  

In general, an ID must find an input $x$ such that $h(x)$ is less than some threshold.  Decreasing this threshold value will increase the difficulty, since one must compute the hash function on more inputs to find an output that is sufficiently small.   We note that many other types of puzzles exist (for examples, see~\cite{Rangasamy2012,karame:low,jerschow:modular}) and our results are likely compatible with other designs.

We assume that each good ID can perform $\mu$ hash-function evaluations per round for $\mu>0$. Additionally, we assume that $\mu$ is of some size polynomial in $n_0$ so that $\log \mu = \Theta(\log n_0)$.  It is reasonable to assume large $\mu$ since, in practice, the number of evaluations that can be performed per second is on the order of millions to low billions~\cite{hashcat,mining,non-spec}.

For any integer $\rho \geq 1$, we define a {\boldmath{$\rho$}}\defn{-round puzzle} to consist of finding $\ell = C\log \mu$ solutions, each of difficulty $\tau = \rho(1-\delta)\mu / (C \log \mu)$, where $\delta>0$ is a small constant and $C$ is a sufficiently large constant depending on $\delta$ and $\mu$.

Let $X$ be a random variable giving the expected number of hash evaluations needed to compute $\ell$ solutions.  Then $X$ is a negative binomial random variable and we have the following concentration bound (see, for example, Lemma 2.2 in~\cite{awerbuch:towards}).

For every $0 < \epsilon \leq 1$, it holds that 
$$Pr(|X - E(X)| \geq \epsilon E(X)) \leq 2e^{- \epsilon^2 \ell / (2(1+\epsilon))}$$

Given the above, we can show that every good ID will solve a $\rho$-round puzzle with at most $\rho\mu$ hash function evaluations, and that the adversary must compute at least $(1-2\delta)d\rho\mu$ hash evaluations to solve every $\rho$-round puzzle. This follows from a union bound over $O(n_0^{\gamma})$ joins and departure events, for $C$ sufficiently large as a function of $\delta$ and $\gamma$.  Note that for small $\delta$, the difference in computational cost is negligible, and that $\mu$ is also unnecessary in comparing costs.  Thus, for ease of exposition, {\bf we assume that each {\boldmath{$\rho$}}-round puzzle requires computational cost {\boldmath{$\rho$}} to solve}.

Finally, each participant $v$ uses a \defn{public key} {\boldmath{$K_v$}}, generated via public key encryption, as its ID.  The input to a puzzle always incorporates  $K_v$ and  $s$, where the \defn{solution string} {\boldmath{$s$}} is selected by $v$ (for good IDs, $s$ will be a random string).


\subsection{How Puzzles Are Used}\label{sec:how-used}

Although, each puzzle is constructed in the same manner, they are used in two distinct ways by our algorithm. First, when a new ID wishes to join the system, it must provide a solution for a $1$-round puzzle; this is referred to as  \defn{\epuzzle}. Here, the input to the puzzle is {\boldmath{$K_v || s || \tstamp$}}, where {\boldmath{$\tstamp$}} is the timestamp of when the puzzle solution was generated. In order to be verified, the value $\tstamp$ in the solution to an entrance puzzle must be within some small margin of the current time.  In practice, this margin would primarily depend on network latency. 


We note that, in the case of a bad ID, this solution may have been precomputed by the adversary by using a future timestamp. This is not a problem since the purpose of this puzzle is only to force the adversary to incur a computational cost at some point, and to deter the adversary from reusing puzzle solutions. Importantly, the \epuzzle is not used to preserve guarantees on the fraction of good IDs in the system.

The second way in which puzzles are used is to limit the fraction of bad IDs in the system; this is referred to as a \defn{\ppuzzle}.  An announcement is periodically made that \emph{all} IDs {\it already in the system} should solve a 1-round puzzle. When this occurs, a \defn{random string} {\boldmath{$r$}} of $\Theta(\log n_0^{\gamma})$ bits is generated and included as part of the announcement. The value $r$ must also be appended to the inputs for all requested solutions in this round; that is, the input is {\boldmath{$K_v || s || r$}}.  These random bits ensure that the adversary cannot engage in a pre-computation attack~\textemdash~where it solves puzzles and stores the solutions far in advance of launching an attack~\textemdash~by keeping the puzzles unpredictable. For ease of exposition, we omit further discussion of this issue and consider the use of these random bits as implicit whenever a \ppuzzle is issued. 

While the same $r$ is used in the puzzle construction for all IDs, we emphasize that a {\it different} puzzle is assigned to each ID since the public key used in the construction is unique. Again, this is only of importance to the second way in which puzzles are used. Using the public key in the puzzle construction also prevents puzzle solutions from being stolen. That is, ID $K_v$ cannot lay claim to a solution found by ID $K_w$ since the solution is tied to the public key $K_w$. 

Can a message $m_v$ from ID $K_v$ be spoofed? This is prevented in the following manner.  ID $K_v$ signs $m_v$ with its private key to get $\texttt{sign}_v$, and then sends $(m_v || \texttt{sign}_v || K_v )$ via \Diffuse. Any other ID can use $K_v$ to check that the message was signed by the ID $K_v$ and thus be assured that ID $K_v$ is the sender. 


\begin{figure}[t!]
\begin{minipage}[h]{1.0\linewidth}
\vspace{1pt}\begin{algorithm}[H]
\smallskip 
\hspace{-7pt}\noindent{}{$\Old$: set of IDs after most recent purge.}\\
\hspace{-7pt}\noindent{}{$S$: current set of IDs}.\\
\hspace{-7pt}\noindent{}{$r$: current random seed}.\\

\smallskip
\hspace{-7pt}\noindent Do the following forever; the committee makes all decisions using Byzantine Consensus:
\begin{enumerate}[leftmargin=6pt]
	\item Each joining ID, $v$, generates a public key; solves an \epuzzle using that key and the current timestamp $\tstamp$ and broadcasts the solution via \Diffuse. Upon verifying the solution, the committee\\ adds $v$ to $\Current$.

	\item If $|(S \cup \Old) - (S \cap \Old)| \geq |\Old|/3$, then  the current committee does the following:\label{s:testStep}
	\begin{itemize}[leftmargin=12pt]
		\item Generate a random string $r$ and broadcast via \Diffuse.
		\item Sets $\Old$ and $\Current$ to the set of IDs that return valid solutions.
		\item Sets new committee  to be the set of IDs submitting the smallest solutions in the range $\left(\frac{1}{2^{(k+1)/d}}, \frac{1}{2^{k/d}} \right]$ for some $k\geq 1$, and diffuses this new committee to all good IDs.
\end{itemize}
\end{enumerate}
\caption{}
\label{alg:ccom}
\end{algorithm}
\end{minipage}
\vspace{-3pt}\caption{Pseudocode for \AlgB.}\label{fig:ccom}
\end{figure}

\subsection{Algorithm Overview}\label{sec:overview}


We begin by describing our algorithm \emph{{\underline{C}}ommensurate {\underline{Com}}putation} (\AlgB). The pseudocode is provided in Figure~\ref{fig:ccom}. A key component of our system is a subset of IDs called a \qcomm for which the \cgoal (recall Section~\ref{sec:our-problems}) must hold; this is proved in Section~\ref{sec:committee-goodness}. 

\subsubsection{Updating and Maintaining System Membership}\label{sec:sys-mem}

The \qcomm tracks the membership in the system using the set {\boldmath{$\Current$}}. Updates to this set are performed in the following manner. Each ID that wishes to join the system must solve an \epuzzle. Each good ID in the \qcomm will check the solution to the \epuzzle and, upon verification, the joining ID is allowed into the system. 

Verification of a puzzle solution requires checking that (1) all $C\log \mu$ inputs to $h$ submitted generate an output that is at most $(C\log \mu)/((1-\delta)\mu)$; and (2) each of these inputs contains the string $r$ and also the individual public key of the ID. Upon verification, the good IDs in the committee solve BC in order to agree on the contents of $\Current$.  Those IDs that fail to submit valid puzzle solutions are denied entrance into the system.

Similarly,  $\Current$ is also updated and agreed upon when a good ID informs the server that it is departing. Of course, bad IDs may not provide such a notification and, therefore, $\Current$ is not necessarily accurate at all times. 

Finally, if we wish to achieve only the Committee Goal, then the last portion of Line 1 is amended: there is no need for the committee to verify solutions or maintain $\Current$.


\subsubsection{Executing Purges} 

At the beginning of the system, the \qcomm knows the existing membership denoted by $\Old$; assume {\boldmath{$|\Old|=n_0$}} initially.  At some point, $|(\Current \cup \Old) - (\Current \cap \Old)| \geq |\Old|/3$. When this happens, Step 2 is immediately triggered, whereby all IDs are issued a  purge puzzle and each ID must respond with a valid solution within $1$ round. The issuing of these \ppuzzle{}s is performed by the \qcomm via the diffusing of $r$ to all IDs; this random string is created via solving BC in order to overcome malicious behavior by bad IDs that are members.\footnote{Agreeing on each random bit of $r$ will suffice. Alternatively, a secure multiparty protocol for generating random values can be used; for example, the result in~\cite{srinathan_pandu_rangan:efficient}.} 

Recall from Section~\ref{sec:model-main} that a round is of sufficient length that the computation time to solve a $1$-round puzzle dominates the round trip communication time between the client and server (we remove this assumption later in Section~\ref{sec:bounded-latency}). The good IDs update and agree on $\Current$ at the end of the purge.   

Figure~\ref{f:overview} illustrates this purging process;  good IDs are black and bad IDs are red.   Let $x$ be the number of IDs after the most recent purge; at most $x/3$ of these IDs are bad, since $\alpha=1/3$. According to Step 2 in \AlgB, the next purge occurs immediately when the number of new IDs entering would become greater than or equal to $x/3$.  At the point when a purge is triggered, the number of new IDs that have been allowed to enter is less than $x/3$.  Thus, the total number of bad IDs is less than $\frac{2}{3}x$, and the total number of good IDs is at least $\frac{2}{3}x$. 


How is $r$ sent such that each good ID trusts it came from the committee? Recall that a  participant $v$ that joins the system generates a public key, $K_v$, and the corresponding private key, $k_v$; the public key is used as its ID. If ID $K_v$ is a committee member, it will sign the random string $r$ using its private key $k_v$ to obtain $\texttt{sign}_v$. Then, $(r || \texttt{sign}_v || K_v)$ is sent using \Diffuse.  Any good ID can verify $\texttt{sign}_v$ via $K_v$ to ensure it returns $r$. This implicitly occurs in Step 2 of Algorithm \ref{fig:ccom}, but we omit the details in the pseudocode for simplicity.

\medskip


\subsubsection{Updating and Maintaining Committee Membership}\label{sec:com-mem}

Recall from Section \ref{sec:model-main} that good IDs always inform all other IDs of their departure. Thus, a good ID $K_v$ that is a committee member will inform other committee members of its departure. Since bad IDs may not inform of their departure, our estimate of churn will not be accurate but this will not jeopardize the committee goal.



\begin{figure}[t]
\centering
\vspace{-3.6cm}  \includegraphics[scale=0.54]{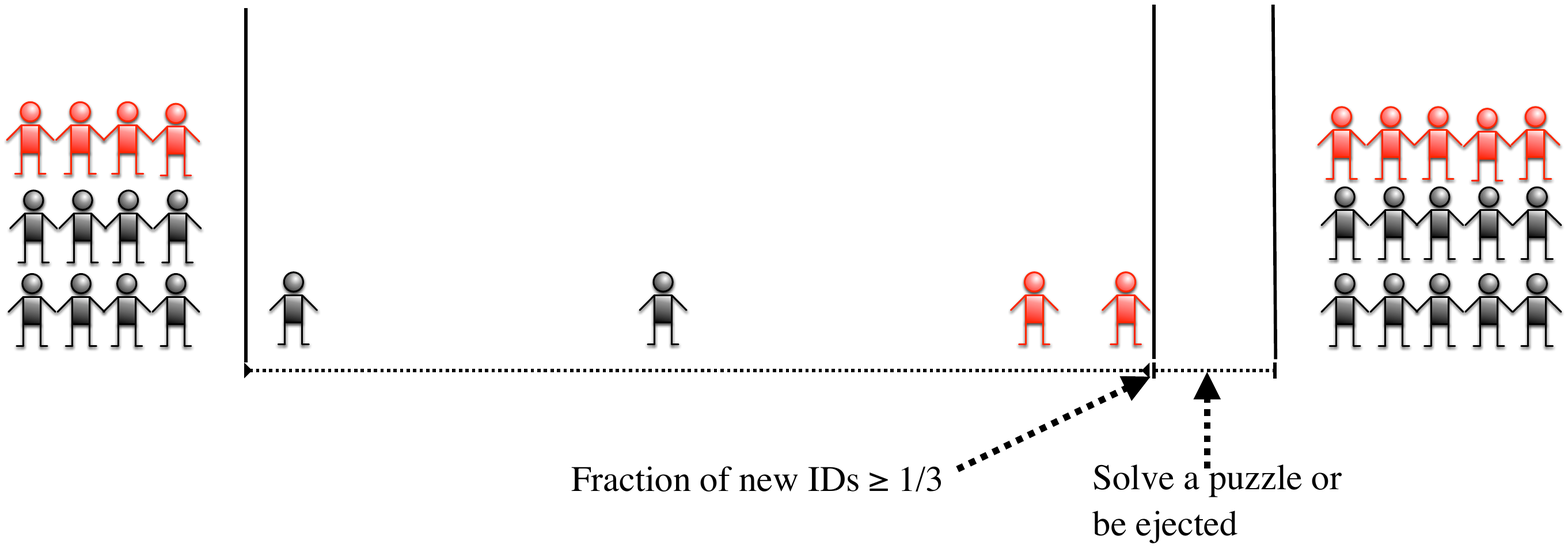}
\vspace{-3.2cm}  \caption{An overview of our approach. The system starts with $x=12$ IDs of which $4$ are bad and $8$ are good. When the $(x/3)^{th} = 4^{th}$ ID joins, a purge is immediately executed and all IDs must solve a \ppuzzle or be ejected from the system. } 
  \label{f:overview} 
\end{figure}

The execution of \AlgB~is conceptually broken into sets of consecutive rounds called \defn{epochs}, which are delineated by the execution of Step 2.  Over the lifetime of the system, the \qcomm~should always satisfy the \cgoal. As joins and departures occur, this invariant may become endangered. To address this issue, \qcomm{}s are disbanded and rebuilt over time. In particular,  a new \qcomm is completed and ready to be used by the start of a each successive epoch, and the currently-used \qcomm is disbanded.

The construction of a new \qcomm occurs in Step 2 over a single round. There is a sequence of \defn{membership intervals} $\left( \frac{1}{2^{(k+1)/d}}, \frac{1}{2^{k/d}} \right]$  for {\boldmath{$k\geq 1$}} and where {\boldmath{$d$}} is a constant we can set subject to $d\geq 20\gamma$ (this is required for the proof of Lemma~\ref{lem:good-committee}). For each interval, the ID that generates the smallest hash function output in the $k^{\mbox{\tiny th}}$ interval becomes the $k^{\mbox{\tiny th}}$ \qcomm~member; this determination is handled by the current \qcomm. In Section~\ref{sec:committee-goodness}, we prove that this process results in a committee of size $\Theta(\log n_0)$, which is a criterion of the \cgoal.

Given that the adversary has a minority of the computational power, we expect the newly-formed committee to have a good majority.  Byzantine consensus allows for agreement on the members of the new committee, and this is communicated to all good IDs in the network via \Diffuse.  These two properties satisfy the other two criteria of the \cgoal; this is argued formally in Section~\ref{sec:committee-goodness}.

Over the round during which this construction is taking place, the existing committee still satisfies the \cgoal; this is proved in Lemma~\ref{lem:good-committee}. After the round completes, the new committee members are known to all good IDs in the system. Those committee members that were present prior to this construction are no longer recognized by the good IDs. Finally, any messages that are received after this single round are considered late and discarded.

\subsubsection{Initialization of Committee}

Finally, we address the initialization of a \qcomm at the time the system is first created.  To do this, we must solve the \emph{view reconciliation} problem.  In this problem, we start out in our model with some unknown set of good IDs, and an adversary that controls an $\alpha$ fraction of the total computational power.  The goal is to ensure that (1) all good IDs agree on a set of IDs; and (2) this set contains all of the good IDs, and a number of bad IDs that is at most an $\alpha$ fraction of the set size.

Several algorithms to solve this problem have been recently described.  In~\cite{andrychowicz2015pow}, Andrychowicz and Dziembowski showed how to solve the problem, even for any fixed $\alpha<1$.  In~\cite{katz2014pseudonymous}, Katz, Miller and Shi described an algorithm to solve the problem when $\alpha<1/2$.  Both of these protocols require time that is $\Theta(n)$, and a number of calls to \Diffuse by good IDs that is $\Theta(n^2)$.  In a recent result~\cite{hou2017randomized}, Hou et al. show how to improve these costs substantially for the case where $\alpha<1/4$.  They describe an algorithm to solve view reconciliation that takes time $\Theta \left(\frac{\log n}{\log \log n} \right)$, and that requires a number of calls to \Diffuse that is $\Theta \left(n\frac{\log n}{\log \log n} \right)$.

Because we assume $\alpha \leq 1/6$, we are able to use the protocol of~\cite{hou2017randomized} for initialization of our system.  We stress that we only need to call this protocol once, at system creation.


\subsection{Analysis of the \cgoal}\label{sec:committee-goodness}
 
In this section, we prove that \whp~\AlgB~preserves the \cgoal over the lifetime of the system. For any epoch $i$, we let {\boldmath{$B_i$}} and {\boldmath{$G_i$}} respectively denote the number of bad and good IDs in the system at the end of epoch $i$, and we let {\boldmath{$N_i = B_i + G_i$}}. Recall that $B_0 < N_0/3$. We also assume that at the time of system creation, the fraction of bad IDs is less than $3/10$; in particular, $B_0 < (3/10)N_0$.

To simplify our presentation, our claims are proved to hold with probability at least $1-\tilde{O}(1/n_0^{\gamma + 2})$, where $\tilde{O}$ hides a poly($\log n_0$) factor.  Of course, we wish the claims of Theorem~\ref{thm:main1} to hold with probability at least $1- 1/n_0^{\gamma+1}$ such that a union bound over $n_0^{\gamma}$ joins and departures yields a w.h.p. guarantee.  By providing this ``slack'' of an $\tilde{\Omega}(1/n_0)$-factor in each of the guarantees of this section, we demonstrate this is feasible while avoiding an analysis cluttered with specific settings for the constants used in our arguments. 


Throughout the following, {\boldmath{$\mathcal{C}_{\mbox{\tiny t}}$}} denotes the membership of \qcomm in round $t\geq0$ of the current epoch.  In particular, {\boldmath{$\mathcal{C}_{\mbox{\tiny 0}}$}} denotes the membership of \qcomm upon its creation at the beginning of the current epoch.  We now prove that the committee goal is met over the duration of any epoch.


Throughout, we assume  $\mu = n_0^{\ell}$ for some constant $\ell \geq 1$, since $\mu$ is a polynomial in $n_0$ (recall Section~\ref{sec:model-main}). The next lemma assumes bounds of $\alpha \leq 1/6$. Although a larger value of $\alpha$ may be tolerable,  this bound is chosen in order to simplify the analysis and provide a clean presentation.

\begin{lemma}\label{lem:good-committee}
Let  $\alpha \leq 1/6$, $d$ be a sufficiently large constant depending on $\gamma$, and $n_0$ be sufficiently large.  Then for any fixed epoch, with probability at least $1-\tilde{O}(1/n_0^{\gamma+2})$, the following holds.  For any round $t\geq0$ in the epoch, $\mathcal{C}_{\mbox{\tiny t}}$ has size $\Theta(\log n_0)$, contains a majority of good IDs, and is known to all good IDs. 
\end{lemma}

\begin{proof} 
Recall that for an ID $K_v$ to become a committee member, it must obtain the {\it smallest} value (via a hash-function evaluation) in $\left(\frac{1}{2^{(k+1)/d}}, \frac{1}{2^{k/d}}\right]$ for some  integer $k\geq 1$. For any $k \geq 1$, let the indicator random variable $X_{v,k} =1$ if ID $K_v$ finds a value in the $k^{\mbox{\tiny th}}$ membership interval, although {\it not} necessarily the smallest value in this interval;  otherwise, $X_{v,k}=0$. 

Fix some an epoch $i$.  Let the set of good IDs present at the beginning of that epoch be denoted by $\mathcal{G}$.  These are the good IDs that may become committee members.  Let the random variable $X_{\mathcal{G},k} = \sum_{K_{v, k}\in \mathcal{G}} X_{v, k}$ which counts the number of values by good IDs that land in the $k^{\mbox{\tiny th}}$ membership interval.

Let the set of all bad IDs present at any point in epoch $i$ be denoted by $\mathcal{B}$. Define $X_{\mathcal{B},k} = \sum_{K_{v, k} \in \mathcal{B}} X_{v, k}$ which counts the number of values by bad IDs that land in the $k^{\mbox{\tiny th}}$ membership interval.

\medskip


\noindent{\underline{\it Expected Value Calculation.}} Fix some $k \geq 1$.  We have:\vspace{-10pt}

$$E[X_{\mathcal{G},k}] = E\left[ \sum_{K_{v}\in \mathcal{G}} X_{v,k}\right] \geq |\mathcal{G}|\frac{\mu}{2^{k/d}} =  G_{i-1}\left(\frac{n_0^\ell}{2^{k/d}}\right)$$ 
\noindent Similarly:\vspace{-4pt}
$$E[X_{\mathcal{B},k}] = E\left[ \sum_{K_v \in \mathcal{B}} X_{v, k}\right] \leq  |\mathcal{B}|\frac{\mu}{ 2^{k/d}} 
\leq \left(\frac{G_{i-1}}{5}\right)\left(\frac{n_0^{\ell}}{ 2^{k/d}} \right) $$ 

\noindent where the last inequality follows from noting that $\alpha = \frac{|\mathcal{B}|}{|\mathcal{B}| + |\mathcal{G}|} \leq 1/6$ implies $|\mathcal{B}| \leq  |\mathcal{G}|/5$.  

\medskip

\noindent\underline{{\it Concentration Bounds on $X_{\mathcal{G}, k}$.}} For a fixed $k \geq 1$, the $X_{v, k}$ variables are independent over different values of $v$.  Thus, a Chernoff bound tightly bounds $X_{\mathcal{G},k}$ such that for any $\delta>0$, we have:\vspace{-3pt}
\begin{eqnarray*}
Pr\left(X_{\mathcal{G}, k} < (1-\delta) \frac{G_{i-1}n_0^{\ell}}{2^{k/d}}\right) & \leq & \exp\left(-\frac{\delta^2 G_{i-1}\, n_0^{\ell}}{2^{(k/d)+1}}  \right)
\end{eqnarray*}
\noindent Let $\eta = G_{i-1}n_0^{\ell}/(\kappa(\gamma+2)\ln n_0)$ for a sufficiently large constant $\kappa>0$.  For the good IDs, we consider $k$ over the range from $1$ to $d\lg\eta$. In this range:\vspace{-3pt}
\begin{eqnarray*}
\exp\left(-\frac{\delta^2 G_{i-1}\, n_0^{\ell}}{2^{(k/d)+1}}  \right)&  \leq & \exp\left(-(\delta^2/2)\kappa(\gamma+2)\ln n_0\right)\\
& = & O\left(1/n_0^{\gamma + 2}\right) 
\end{eqnarray*}
\noindent for $\kappa\geq 2/\delta^2$. Taking a union bound over all $d\lg\eta$ intervals, the following is true with probability at least $1-\tilde{O}(1/n_0^{\gamma+2})$. For all $1 \leq k\leq d\lg\eta$, the number of good IDs in the $k^{th}$ interval is at least $(1-\delta) G_{i-1}\mu/2^{k/d}$. 

A similar calculation yields upper bounds on $X_{\mathcal{G}, k}$ with the same probabilistic guarantee for each interval over this range of $k$:\vspace{-3pt}
\begin{eqnarray*}
Pr\left(X_{\mathcal{G}, k} > (1+\delta) \frac{G_{i-1} n_0^{\ell}}{2^{k/d}}\right) & = & O\left(1/n_0^{\gamma + 2}\right) 
\end{eqnarray*}
 
Taking a union bound over over all $d\lg\eta$ intervals, the following is true with probability at least $1-\tilde{O}(1/n_0^{\gamma+2})$: for all $1 \leq k\leq d\lg\eta$, the number of puzzle solutions from good IDs in the $k^{th}$ interval is at most $(1+\delta') G_{i-1}\mu/2^{k/d}$.

\medskip\medskip

\noindent\underline{{\it Concentration Bounds on $X_{\mathcal{B}, k}$.}} Over the range of $1 \leq k\leq d\lg\eta$, a similar argument proves that: \vspace{-5pt}

 $$(1 - \delta)\left(\frac{G_{i-1}}{5}\right)\left(\frac{n_0^{\ell}}{2^{k/d}}\right) \leq X_{\mathcal{B},k} \leq (1 + \delta)\left(\frac{G_{i-1}}{5}\right)\left(\frac{n_0^{\ell}}{2^{k/d}}\right)$$
 
\noindent with probability at least $1-O(1/n_0^{\gamma+2})$.  A similar union bound shows that the above holds with probability $1-\tilde{O}(1/n_0^{\gamma+2})$ for all $1 \leq k\leq d\lg\eta$.

However, can the adversary obtain values in membership intervals for $k>d\lg(\eta)$? For such intervals, we expect only a small number of values and we pessimistically assume that, if one exists, it belongs to the adversary rather than a good ID.

Each such membership interval has size less than $1/(2^j \eta)$ for $j\geq 1$. We wish to know a value $j$ such that, with probability $1-O(1/n_0^{\gamma})$, the adversary does not obtain a puzzle solution in the corresponding interval, nor in any of the subsequent (smaller) intervals.  This will bound the number of committee members the adversary obtains from the membership intervals for $k>d\lg(\eta)$. We solve for $j$ in $n_0^{\ell}/(2^j \eta) \leq 1/n_0^{\gamma+2}$ which yields
\begin{eqnarray*}
j &\geq& (\ell + \gamma+2) \lg n_0 - \lg \eta\\
 &=& (\gamma+2)\lg n_0 - \lg(G_{i-1}) + \lg(\kappa(\gamma+2)\ln n_0)  
\end{eqnarray*}
Since $G_{i-1} \geq n_0$, and for sufficiently large $n_0$, it follows that $\lg(G_{i-1}) \geq  \lg(\kappa(\gamma+2)\ln n_0)$. Therefore, with probability at least $1-O(1/n_0^{\gamma+2})$, the adversary gains at most $(\gamma+2) \lg n_0$ additional committee members.  

\medskip


\noindent\underline{{\it Size of Committee.}}  We note that the same argument used above to bound the number of bad IDs for membership intervals with $k>d\lg(\eta)$ also applies to good IDs. By this fact, and the above bounds, with probability at least $1-\tilde{O}(1/n_0^{\gamma+2})$ the following holds:\vspace{-5pt} 

$$ d\lg(\eta)  \leq |\mathcal{C}_{\mbox{\tiny 0}}| \leq d\lg(\eta) + 2(\gamma+2)\lg(n_0)$$

\noindent By the definition of $\eta$, and given that $\ell$ and $\gamma$ are constants, it follows that $ |\mathcal{C}_{\mbox{\tiny 0}}| =  \Theta(\log n_0)$. Finally, since the committee size can only decrease, and further it decreases by at most a $1/3$-factor within an epoch, it follows that  $|\mathcal{C}_{\mbox{\tiny t}}| = \Theta(\log n_0)$ for any round $t > 0$.

\medskip
\medskip

\noindent\underline{{\it  Fraction of Good Committee IDs at the Beginning of Epoch.}} We begin by analyzing the fraction of good IDs in the committee at the beginning of the epoch. Each hash-function evaluation that falls in the $k^{\mbox{\tiny th}}$ membership interval has the same probability as any other of being the smallest. Therefore, pulling the above arguments together: with probability at least $1-\tilde{O}(1/n^{\gamma+2})$, the probability that a good ID obtains the smallest output in the $k^{\mbox{\tiny th}}$ interval for $1\leq k \leq d\lg(\eta)$ is at least:\vspace{-12pt}

\begin{eqnarray*}
\frac{  (1 - \delta) G_{i-1} n_0^{\ell}/2^{k/d}   }{  (1 + \delta) G_{i-1}n_0^{\ell}/2^{k/d} +  (1 + \delta)(G_{i-1}/5) n_0^{\ell}/ 2^{k/d}  } \hspace{-5pt} & \geq & \hspace{-5pt} (1-\epsilon')\frac{5}{6}
\end{eqnarray*}

\noindent for any $\epsilon' > 0$ provided that $\delta$ is sufficiently small. 

To obtain a lower bound on the number of good IDs in $\mathcal{C}_{\mbox{\tiny 0}}$, define the indicator random variable $Y_k=1$  if some good ID has the smallest output in the $k^{\mbox{\tiny th}}$ interval; otherwise, $Y_k=0$.  Let $Y = \sum_{k=1}^{d\lg(\eta)} Y_k$.  This implies a lower bound on the expected number of good IDs: \vspace{-12pt}

\begin{eqnarray*}
E[Y] \hspace{-5pt} &= & \hspace{-8pt} E\left[\sum_{k=1}^{d\lg(\eta)} Y_k\right]  = \sum_{k=1}^{d\lg(\eta)} E[Y_k]\\
\hspace{-5pt} &\geq& \hspace{-8pt}\sum_{k = 1}^{d\lg(\eta)} (1-\epsilon') (5/6)= (1-\epsilon')(5/6)d\lg(\eta)
\end{eqnarray*}


We now make two observations. First, the $Y_k$ values are independent, thus using a Chernoff bound, we can obtain a tight bound on the number of good IDs in  $\mathcal{C}_{\mbox{\tiny 0}}$:

\begin{eqnarray*}
Pr \left( Y < {(1-2\epsilon')}(5/6)d\lg(\eta)\right) \hspace{-5pt} & =& Pr \left( Y < \left(1-\frac{\epsilon'}{1-\epsilon'}\right)E(Y)\right) \\
&\leq& \hspace{-8pt} \exp\left({-\frac{\epsilon'^2}{2(1-\epsilon')}\left(\frac{5\,d \lg \eta}{6}\right)}\right)\\
\hspace{-5pt} &= & \hspace{-8pt}O(1/n_0^{\gamma + 2})
\end{eqnarray*}

where the last equality holds for $d \geq \frac{(1-\epsilon')12(\gamma + 2)}{\epsilon'^2 5\lg(e)}$.


Second, we derived above that, with probability at least $1-\tilde{O}(1/n^{\gamma+2})$, $|\mathcal{C}_{\mbox{\tiny 0}}| \leq d\lg(\eta) + 2(\gamma+2)\lg n_0$. Therefore, using our first observation, with probability at least $1-\tilde{O}(1/n^{\gamma+2})$,  the fraction of good IDs in $\mathcal{C}_{\mbox{\tiny 0}}$ is at least: 
\begin{eqnarray*}
\frac{(1-2\epsilon')(5/6)d\lg(\eta)}{d\lg(\eta) + 2(\gamma+2)\lg(n_0)} \geq \frac{(1-2\epsilon')(5/6)d\lg(\eta)}{d\lg(\eta) + 2(\gamma+2)\lg(\eta)}  \geq 7/10\\
\end{eqnarray*}

\noindent where the second inequality holds for any $d \geq \frac{21(\gamma+2)}{(2-25\epsilon')}$. Therefore, for sufficiently large $d$,  $\mathcal{C}_{\mbox{\tiny 0}}$ has at least a $7/10$-fraction of good IDs.

The existing committee members come to agreement on $\mathcal{C}_{\mbox{\tiny 0}}$ by executing a Byzantine Consensus protocol that succeeds with probability at least $1-O(1/n_0^{\gamma+2})$. Each committee member propagates its view of  $\mathcal{C}_{\mbox{\tiny 0}}$ to all remaining IDs in the network via \Diffuse.  A good ID not in the committee will take the majority value of these different views and, since the committee has a good majority, the majority value will be $\mathcal{C}_{\mbox{\tiny 0}}$. Therefore, the newly-formed $\mathcal{C}_{\mbox{\tiny 0}}$ is known to all good IDs.\medskip

\noindent\underline{{\it  Fraction of Good Committee IDs at Any Point in the Epoch.}} What about the fraction of good committee IDs over the entire epoch? Over the epoch, at most $|\mathcal{C}_{\mbox{\tiny 0}}|/3$ IDs can depart. In the worst case, these are all good IDs. 
Let $X$ denote the number of good IDs in $\mathcal{C}_{\mbox{\tiny 0}}$, let $Y \geq X$ denote the total number of IDs in $\mathcal{C}_{\mbox{\tiny 0}}$, and let $Z \leq Y/3$ be the number of good IDs that have left.
Then, $\frac{X - Z}{Y - Z} \geq \frac{X - (Y/3)}{(2/3)Y} = (3/2)(X/Y) - (1/2) \geq (3/2)(7/10) - (1/2) = 11/20$.  Where $X/Y \geq 7/10$ holds with probability at least $1-\tilde{O}(1/n_0^{\gamma+2})$ from the argument above. 

Recall that at most a constant $\epsilon_0$-fraction of good IDs may depart over the single round in which the new committee is being formed during Step 2 of \AlgB (see our model for joins and departures described in Section~\ref{sec:model-main}). By a Chernoff bound, for $\epsilon_0<1/20$, the fraction of good IDs that leave the committee is tightly bounded such that the fraction of good IDs remaining in the committee exceeds $1/2$ with probability at least $1-O(1/n_0^{\gamma+2})$ during this final round. Therefore, for the current epoch, a majority of good IDs exists in $\mathcal{C}_{\mbox{\tiny t}}$ until a new committee has been created. \medskip

\noindent\underline{{\it  Conclusion.}} By the above steps of our argument, with probability $1-\tilde{O}(1/n_0^{\gamma+2})$, over all rounds $t\geq 0$ in a fixed epoch, the following properties hold for $\mathcal{C}_{\mbox{\tiny t}}$: it has size $\Theta(\log n_0)$; it contains a majority of good IDs; and it is known to all good IDs in the system.
\end{proof}


\subsection{Analysis of the \sgoal}\label{sec:full-dynamism}

\begin{lemma}\label{lem:bound_b}
For all $i \geq 0$,  $B_i < N_i/3$.
\end{lemma}
\begin{proof}
This is true by assumption for $i=0$.  For $i>0$, note that the adversary has computational power less than $G_i/2$.  Thus, after Step 2 that ends epoch $i$, we have $B_i < G_i/2$. Adding $B_i/2$ to both sides of this inequality yields $(3/2) B_i < N_i/2$, from which, $B_i < N_i/3$.
\end{proof}

Let {\boldmath{$n_i^a, g_i^a, b_i^a$}} denote the total, good, and bad IDs that arrive over epoch $i$. Similarly, let  {\boldmath{$n_i^d, g_i^d, b_i^d $}} denote the total, good, and bad IDs that depart over epoch $i$. 

Note that the \qcomm will always have an accurate value for all of these variables except for possibly $b_i^d$ --- recall from Section~\ref{sec:model-main} that bad IDs do not need to give notification when they depart --- and, consequently, $n_i^d$; in these two cases, the \qcomm may hold values which are underestimates of the true values. 

\begin{lemma}\label{lem:badlesshalf-full}
The fraction of bad IDs is always at most $1/2$. 
\end{lemma}
\begin{proof}
Fix some epoch $i>0$.  Recall that an epoch ends when $|(\Current \cup \Old) - (\Current \cap \Old)| \geq |\Old|/3$ where $|\Old| = N_{i-1}$. Therefore, we have $b_i^a + g_i^d \leq N_{i-1}/3$.  We are interested in the maximum value of the ratio of bad IDs to total IDs at any point during the epoch.  Thus, we pessimistically assume all additions of bad IDs and removals of good IDs come first in the epoch.  We are then interested in the maximum value of the ratio:
$$\frac{B_{i-1} + b_i^a}{N_{i-1} + b_i^a - g_i^d}.$$  
By Lemma~\ref{lem:bound_b}, $B_{i-1} \leq N_{i-1}/3$.
 Thus, the we want to find the maximum of $\frac{N_{i-1}/3 + b_i^a}{N_{i-1} + b_i^a - g_i^d}$, subject to the constraint that $b_i^a + g_i^d \leq N_{i-1}/3$.  
This ratio is maximized when the constraint achieves equality, that is when $g_i^d = N_{i-1}/3 - b_i^a$.  Plugging this back into the ratio, we get
\begin{eqnarray*}
\frac{N_{i-1}/3 + b_i^a}{N_{i-1} + b_i^a - g_i^d}
& \leq & \frac{N_{i-1}/3 + b_i^a}{2N_{i-1}/3 + 2b_i^a}\\
& = & 1/2
\end{eqnarray*}
Therefore, the maximum fraction of bad IDs is $1/2$ at any point during any arbitrary epoch $i>0$.

Finally, we note that this argument is valid even though $\Current$ may not account for bad IDs that have departed {\it without} notifying the \qcomm (recall this is possible as stated in Section~\ref{sec:model-main}). Intuitively, this is not a problem since such departures can only lower the fraction of bad IDs in the system; formally, the critical equation in the above argument is $b_i^a + g_i^d \leq N_{i-1}/3$, and this does not depend on $b_i^d$.
\end{proof}


\subsection{Cost Analysis and Proof of Theorem~\ref{thm:main1}}

We now examine the cost of running \AlgB.  When computing the bandwidth cost for the decentralized network, we assume that the good IDs are connected in a sparse overlay network.  In particular, for epoch $i$, we assume that the number of edges in this network is $\tilde{O}(G_i)$ where $\tilde{O}$ indicates the omission of a poly($\log G_i$) = poly($\log n_0$) factor.  

\begin{lemma}\label{lemma:committee-cost}  Let $T_C$ and $T_B$ denote the total computational and bandwidth costs, respectively, incurred by the adversary. Let $\gnew$ denote the number of good IDs that have joined the system.  Then,  \AlgB has costs as follows: 
	\begin{itemize}
		\item The total computational cost to the good IDs is $O\left(T_C+ \gnew \right)$. 
       		\item The total bandwidth cost to the good IDs is  $\tilde{O}\left(T_B + \gnew \right)$.
	\end{itemize}
\end{lemma}
\begin{proof} 
We begin with computational cost.  In epoch $i$, the computational cost to the good IDs equals $G_i$, which is no more than $G_{i-1} + g^a_i$. By the end of epoch $i$, we know that a $1/3$ fraction of the IDs in $\mathcal{C}_{old}$ leave.  Let $Cost_i$ be the computational cost to the algorithm in epoch $i$, and let $T_i$ be the computation cost spent by the adversary in epoch $i$.  There are two cases. \medskip


\noindent
Case 1: At least a $1/6$ fraction of the IDs that leave the committee are good.  Since the good IDs in the committee are a logarithmic-sized random sample of the good IDs in the system, $g_i^d = \Theta(G_{i-1})$, \whp~by a standard Chernoff bound.   Thus, $Cost_i \leq c_1 g_i^d + g_i^a$ for some constant $c_1$. \medskip

\noindent
Case 2: At least a $1/6$ fraction of the IDs that leave the committee are bad.  Then, at the beginning of the epoch, the adversary must have incurred $\Theta(G_{i-1})$ computational cost in order to obtain positions in the committee for this constant fraction of bad IDs.  Thus, $T_{i-1} = \Theta(G_{i-1})$, and so $Cost_i \leq c_2 T_{i-1} + g_i^a$ for some constant $c_2$. \medskip

Combining these two cases, we have that for every epoch $i$, it is the case that $Cost_i \leq c_1 g_i^d + c_2 T_{i-1} + g_i^a$.  By summing over all epochs, we get that the total computational cost to the algorithm is no more than the following.

\begin{eqnarray*}
	\sum_i (c_1 g_i^d + c_2 T_{i-1} + g_i^a) & \leq & c_2 \sum_i  T_i + c_1\sum_i  g_i^d  +  \sum_i g_i^a\\
	& = & O\left(T_C+ \gnew \right) 
\end{eqnarray*}
\noindent where the last line follows since $\gnew  = \sum_i g_i^a = \Omega(g_i^d)$.


To analyze the bandwidth cost, first recall that we assume the IDs are connected in a network with $\tilde{O}(G_i)$ edges.  A record-breaking analysis shows that each edge in this network will be traversed by an expected $O(\log G_i)$ messages during the formation of a new \qcomm.  In particular, the source of each edge perceives a random sequence of $O(G_i)$ values, and in expectation $O(\log G_i)$ of these will be the smallest seen so far, and so will require a message.  Thus, the number of messages sent equals the number of minimum values seen in a sequence of $x$ independent and uniformly distributed random variables on the interval $(0,1]$.  This expectation is $O(\log x)$; see~\cite{grinstead2012introduction} for details.  Thus the total bandwidth cost during the formation of a new committee is $\tilde{O}(G_i)$.  

Since this is within logarithmic factors of the computational cost, we can apply the same analysis as for the computational cost to achieve the bound given in the lemma. 
\end{proof}

\noindent Pulling the pieces together from the previous sections, we can now prove Theorem~\ref{thm:main1}.
\begin{proof}
By Lemma~\ref{sec:committee-goodness}, with probability at least $1-O(1/n_0^{\gamma+1})$, the \qcomm has a majority of good IDs, has size $\Theta(\log n_0)$, and its members are communicated to all good IDs via \Diffuse; therefore, the \cgoal is met. Given that the \cgoal holds, by Lemma~\ref{lem:badlesshalf-full} the fraction of bad IDs is less than $1/2$ over the lifetime of the system; therefore, the \sgoal is met.  The computational cost and bandwidth cost follows directly from Lemma~\ref{lemma:committee-cost}. 
\end{proof}



\subsection{Enhancements to \AlgB}\label{sec:enhancements}

In this section, we present two enhancements to our main result. The first allows for a reduction in the amount of state that needs to be maintained by \qcomm members. The second illustrates how our results continue to hold when there is bounded communication latency.




\subsection{Reducing Cost of Tracking Departures}\label{sec:reducing-state}

In order to determine when an epoch should end, \AlgB implicitly requires that the committee communicate with each good ID to learn of its departure; recall our assumption in Section~\ref{sec:model-main} that good IDs alert the committee of their departure. Additionally, the committee must also calculate the symmetric difference between two (possibly large sets) $\Current$ and $\Old$. Unfortunately, both this tracking and calculation can be expensive.

To reduce this cost, we modify \AlgB such that the committee tracks the state of $O(\log n_0)$ IDs currently in the system; this set is referred to as a \defn{sample set}. By tracking how many IDs depart from the (small) sample set, the committee is able to determine the end of an epoch with less computational cost. We refer to this modified algorithm as \textsc{Enhanced}~\AlgB~(\AlgBRS).

The amount of churn experienced by the system is estimated through sample sets $\sOld$ and $\sCurrent$, which are generated by the committee using a hash function {\boldmath{$h'$}} computed using secure multi-party computation (MPC) such as in \cite{applebaum2010secrecy, beerliova2006efficient, bogetoft2009secure, damgard2006scalable, damgaard2008scalable, dani2014quorums, goldreich1998secure}. Note that $h'$ is unknown to the IDs and thus, the adversary cannot precompute the sample sets. The committee adds an ID, say $v$ to the sample set if the hash of its ID $h'(v) \leq \frac{c\log{n_0}}{|\Old|}$, for some constant $c>0$. A fresh sample set, {\boldmath{$\sOld$}}, is generated at the beginning of every epoch. As new IDs join the system, they are sampled by the committee using $h'$ and added to the current sample set, {\boldmath{$\sCurrent$}}. 

Initially, the committee knows the existing membership denoted by $\Current$, and so  $|\Current|=n_0$.  The committee generates $\sOld$ and updates the current sample set, $\sCurrent$ as described above. At any time during the execution, if $|(\sCurrent \cup \sOld) - (\sCurrent \cap \sOld)|$ exceeds $\sOld/4$, the system undergoes a purge test (Step 2). As before, the committee broadcasts $r$ to conduct \ppuzzle with a round of sufficient length that the computation time dominates the round trip communication time between the client and committee.Those IDs that fail to submit valid puzzle solutions during step 2 are de-registered --- that is, they are effectively removed from the system.  


\renewcommand*{\algorithmcfname}{Enhanced Commensurate Computation (\AlgBRS)}
\begin{figure}[t!]
\begin{minipage}[h]{1.0\linewidth}
\vspace{1pt}\begin{algorithm}[H]
\smallskip 
\hspace{-7pt}\noindent{}{$\sOld$:  sample set of IDs after most recent purge.}\\
\hspace{-7pt}\noindent{}{$\sCurrent$: current sample set of IDs}.\\
\hspace{-7pt}\noindent{}{$\Old$: set of IDs after most recent purge.}\\
\hspace{-7pt}\noindent{}{$\Current$: current set of IDs}.\\
\hspace{-7pt}\noindent{}{$r$: current random seed}.\\

\medskip
\hspace{-7pt}\noindent Do the following forever; the committee makes all decisions using Byzantine Consensus:
\begin{enumerate}[leftmargin=6pt]
	\item Each joining ID, $v$, generates a public key; solves an \epuzzle using that key and the current timestamp $\tstamp$ and broadcasts the solution via \Diffuse. Upon verifying the solution $s_v$, the committee\\ adds $v$ to $\Current$.  
	
	\item If $h'(v) \leq \frac{c\log{n_0}}{|\Old|}$, then add $v$ to $\sCurrent$.
	
	\item Ping all IDs in $\sOld \cup \sCurrent$ in every round, and remove those IDs that do not respond.

	\item If $|(\sCurrent \cup \sOld) - (\sCurrent \cap \sOld)| \geq |\sOld|/4$, then the current committee does the following:
	\begin{itemize}[leftmargin=12pt]
		\item Generate a random string $r$ and broadcast via \Diffuse.
		\item Sets $\Old$ and $\Current$ to the set of IDs that return valid solutions.
		\item Sets $\sOld$ and $\sCurrent$ to the set of IDs that return a valid solutions whose value under $h'$ is at most $ \frac{c\log{n_0}}{|\Old|}$
		\item Sets the current committee to be the set of IDs submitting the smallest solutions in the range $\left(\frac{1}{2^{(k+1)/d}}, \frac{1}{2^{k/d}} \right]$ for some $k\geq 1$, and diffuses this new committee to all good IDs.
	\end{itemize}
\end{enumerate}
\caption{}
\end{algorithm}
\end{minipage}
\vspace{-3pt}\caption{Pseudocode for \AlgBRS.}\label{alg:ccomrs}
\end{figure}


\subsubsection{Estimating the Symmetric Difference}

\noindent For an epoch $i$, suppose $\Old$ is the set of IDs in the system at the beginning of the epoch and $\Current$ is the set of IDs in the system at time $t$. Let {\boldmath{$f_N^i$}} denote the fraction of new IDs that remain active in the system through epoch $i$, and let {\boldmath{$f_L^i$}} denote the fraction of IDs that leave from $\Old$ in epoch $i$. Similarly, let {\boldmath{$f_N^{i'}$}}  denote the fraction of new IDs that remain active in $\sCurrent$ through epoch $i$, and {\boldmath{$f_L^{i'}$}}                   denote the fraction of IDs that leave from $\sOld$ in epoch $i$.

\begin{lemma}\label{lem:samp_bound}
For any epoch $i$, if $f_N^{i'} + f_L^{i'} \leq \frac{1}{4}$, then $f_N^i + f_L^i \leq \frac{1}{3}$ w.h.p.
\end{lemma}
\begin{proof}
The probability that an ID becomes member of the sample set is $\frac{c\log{n_0}}{|\Old|}$. Fix a time step during an epoch.  Let $\Current$ be the set of IDs in the system
at that time step.  Then the size of the sample set will be $|\Current|\frac{c\log{n_0}}{|\Old|} \geq \frac{2}{3}c\log{n_0}$ in expectation. From Section 4.2.3 on Chernoff Bounds of \cite{mitzenmacher2017probability}, for constants $\delta = \frac{1}{24}$ and $c \geq 1728\gamma$, the following hold:
$$Pr(f^i_L \notin \lbrack f_L^{i'} - \delta,f_L^{i'} + \delta\rbrack ) < e^{\frac{-c\delta^2\log{n_0}}{2}} + e^{\frac{-c\delta^2\log{n_0}}{3}} $$
$$Pr(f^i_N \notin \lbrack f_N^{i'} - \delta,f_N^{i'} + \delta\rbrack ) < e^{\frac{-c\delta^2\log{n_0}}{2}} + e^{\frac{-c\delta^2\log{n_0}}{3}} $$ 
Thus, w.h.p. $f^i_L \in \lbrack f_L^{i'} - \delta,f_L^{i'} + \delta\rbrack$ and $f^i_N \in \lbrack f_N^{i'} - \delta,f_N^{i'} + \delta\rbrack$, using which we get
\begin{align*}
f_L^i + f_N^i &\leq f_L^{i'} + f_N^{i'} + 2\delta
\end{align*}
On solving the above inequality using the fact that $f_N^{i'} + f_L^{i'} \leq \frac{1}{4}$, we get $f_N^i + f_L^i \leq \frac{1}{3}$.
\end{proof}

\begin{lemma}{\label{trigBound}}
For any epoch $i$, if $f_N^{i'} + f_L^{i'} \geq \frac{1}{4}$, then $f_N^{i} + f_L^{i} \geq \frac{1}{6}$ w.h.p.
\end{lemma}
\begin{proof}
Following the same argument as Lemma \ref{lem:samp_bound}, we have $f_N^i \geq f_N^{i'} - \delta$ and $f_L^i \geq f_L^{i'} - \delta$. On summing the two inequalities, we get:
\begin{align*}
f_N^i + f_L^i &\geq f_N^{i'} - \delta + f_L^{i'} - \delta \geq \frac{1}{4} - 2\delta \geq \frac{1}{6}
\end{align*}
\end{proof}

\begin{theorem}{\label{theoRD}}
Assume that the adversary holds at most an $\alpha = 1/6$ fraction of the total computational power of the network. Then, w.h.p., algorithm \AlgBRS ensures the following properties hold  over the lifetime of the system:\vspace{5pt}
\begin{itemize}[topsep=0pt,parsep=0pt,partopsep=0pt,leftmargin=13pt,labelwidth=6pt,labelsep=4pt]
\item The \sgoal and \cgoal are maintained.
\item The cumulative computational cost to the good IDs is $O(T_C + g_{new})$.
\item The cumulative bandwidth cost to the good IDs is $\tilde{O}(T_B + g_{new})$.
\end{itemize}
\end{theorem}

\begin{proof}
From Lemma \ref{lem:samp_bound}, when the purge puzzle is triggered in   \AlgBRS, the system has seen at most $1/3$ churn. Hence, from Lemma \ref{lem:badlesshalf-full}, the Population Goal is satisfied. Also, since the sample set enhancements do not change the rules for forming committees, the Committee Goal is also satisfied.

For calculating the cumulative computational cost to the good IDs, note from Lemma \ref{trigBound} that when the purge puzzle is triggered, the fraction of churn in the system is at least $\frac{1}{6}$, which is half of the triggering fraction in \AlgB. Thus, extending the results of Theorem \ref{thm:main1} by multiplying the cumulative computational cost by $2$, we get that the cumulative computational cost to the good IDs is $O(T + g_{new})$ for \AlgBRS. 

The cumulative bandwidth cost follows from Lemma \ref{lemma:committee-cost}.
\end{proof}


\subsection{Handling Bounded Communication Latency}\label{sec:bounded-latency}

Previously, all messages were assumed to be delivered with negligible delay when compared with the time to solve computational puzzles. In this section, we relax this assumption to the following:   \textit{the adversary may choose to arbitrarily delay messages, subject to the constraint that they are delivered within a constant number of rounds $\Delta>0$. } Such a model of communication delay is often referred to as \defn{partial synchrony} or {\boldmath{$\Delta$}}\defn{-bounded delay} in the literature \cite{dwork1988consensus,pass2017analysis}. We address this issue in the following analysis for each of our algorithms.

\subsubsection{In \textsc{Distributed} \AlgB}

In \AlgB, messages are exchanged with good IDs every time the committee is constructed. Thus, we increase the wait time for receiving the solutions to the \ppuzzle by $2\Delta$.  To simplify the analysis of a $\Delta$-bounded delay in the network where the adversary holds  an $\alpha$-fraction of the total computational power, we can instead consider a $0$-bounded delay network, but where the adversary holds an increased fraction of the total computational power; this is addressed by Lemma \ref{boundedfrac}. 

\begin{lemma}{\label{boundedfrac}}
A $\Delta$-bounded network latency with an adversary that holds an $\alpha$-fraction of total computational power of the network is equivalent to a $0$-bounded network latency with an adversary that holds a $\frac{2\alpha\Delta+\alpha}{2\alpha\Delta + 1}$-fraction of the total computational power of the network.
\end{lemma}
\begin{proof}
Recall from Section~\ref{sec:model-main}, the fraction of computational power with an adversary is $\alpha$-times the total computational power of the network. Given the $\Delta$-bounded delay in the network, the adversary has $2\Delta+1$ rounds to perform evaluations. Fix a time time step during an epoch.  Suppose $P$ is the total computational power of the network at that time step. Then, the number of evaluations of the hash function an adversary can make is $P\alpha(2\Delta + 1)\mu$. 

Note that the computational power of the good IDs remains the same, which is $P(1-\alpha)\mu$ since they obey the protocol, and we assume that the bounded delay only benefits the adversary. Thus, the $\Delta$-bounded delay model is equivalent to the $0$-bounded delay model with an adversary who holds a $\frac{2\alpha\Delta+\alpha}{2\alpha\Delta + 1}$-fraction of the computational power.
\end{proof}



\begin{lemma}{\label{bounded-dist}}
Assume a $\Delta$-bounded delay and an adversary that holds an $\alpha$-fraction of computational power, for $\alpha \leq 1/(10\Delta + 6)$. With probability at least $1-O(1/n_0^{\gamma+1})$, the fraction of good IDs in the committee exceeds $1/2$ for the duration of the epoch.
\end{lemma}
\begin{proof}
From Lemma~\ref{boundedfrac}, the $\Delta$-bounded delay is equivalent to a $0$-bounded delay with  an adversary that holds a $\frac{2\alpha\Delta+\alpha}{2\alpha\Delta + 1}$-fraction of computational power of the network. By Lemma~\ref{lem:good-committee}, for a $0$-bounded delay model  with an adversary that holds an $\alpha'$-fraction of the computational power for $\alpha'\leq 1/6$, the fraction of good IDs exceeds $7/10$. Substituting $\alpha' = \frac{2\alpha\Delta+\alpha}{2\alpha\Delta + 1}$, the $\Delta$-bounded delay model with an $\alpha$-fraction of computational power has at least $7/10$-fraction of good IDs, for $\alpha \leq 1/(10\Delta+6)$.
\end{proof}

\subsubsection{In \AlgBRS}

In \AlgBRS, the purge puzzle is triggered when $|(\sCurrent \cup \sOld) - (\sCurrent \cap \sOld)| \geq |\sOld|/4$. Due to $\Delta$-bounded network latency, the adversary can choose to delay the departure notifications from the good IDs by $\Delta$ rounds, thereby delaying the detection of the need to purge. This may jeopardize the Population Goal and we address this issue here.

Recall from Section~\ref{sec:model-main} that the maximum fraction of good IDs that can depart in a round is $\epsilon_0$. Hence, the maximum fraction of good IDs that can depart in $\Delta$-rounds is $\Delta\epsilon_0$. In order to satisfy the Population goal, it is required that purge puzzles be issued when  $|(\sCurrent \cup \sOld) - (\sCurrent \cap \sOld)| \geq |\sOld|(1/4 - \epsilon_0 \Delta)$. Due to this modification, our computational and bandwidth costs will be increased by at most a constant factor. 


\section{Empirical Validation}\label{sec:empirical}

We simulate \AlgB~to evaluate its performance against a well-known~\POW~scheme named \AlgA~\cite{li:sybilcontrol}. The experiments allow us to (1) compare the computational cost incurred by the good IDs and (2) investigate the extent to which a Sybil Attack can impact the fraction of good IDs under these two algorithms.  


Our goal is to offer a proof of concept, not a full-fledged comparison. To this end, we simulate a centralized version of \AlgB from our preliminary work in~\cite{pow-without}.  In this version, the \qcomm is replaced by a single server that assumes the duties of the committee; the server unilaterally generates the random string $r$ and verifies solutions to entrance and purge puzzles (the use of membership intervals in Step 2 of \AlgB is no longer needed; recall Section~\ref{sec:com-mem}). This approach greatly reduces the number of parameters that we need to correctly estimate for our simulations, and it allows us to focus on obtaining a first-order approximation of the system. Because the distributed version has asymptotically equal resource costs, we expect similar performance for the distributed algorithm.

\noindent{\it Overview of \AlgA.} Under this algorithm, the number of bad IDs are limited through the use of computational puzzles. Each ID must solve a puzzle to join the system. Additionally, each ID tests its neighbors every 5 seconds with a puzzle, removing those IDs from its list of neighbors that failed to provide a solution within a time limit. It may be the case that an ID is a neighbor to more than one ID and thus, receives multiple puzzles to solve simultaneously; in this case, they are combined into a single puzzle whose solution satisfies all the received puzzles. We implement our own simulation of the \AlgA~algorithm. 

This comparison is not a straw man.  \AlgA, with its perpetual resource burning, is typical of existing PoW-based schemes; see Section~\ref{sec:related-work} for examples. One may consider enhancing \AlgA~via anomaly-detection techniques to determine when to use puzzle challenges.  However, bad IDs can hide by obeying protocol until an opportunity to attack presents itself.  Thus, such an approach is unlikely to provide acceptable security guarantees.

\medskip
 
\noindent{\it Points of Comparison.} To compare both algorithms fairly, we assume that the computational cost for solving a single \POW~is $1$ for both algorithms. Since in both algorithms require that a new ID solves a puzzle~to join the system, we refrain from measuring this computational cost. We let the fraction of computational power of the adversary, $\frac{g}{\alpha}$ be the same for both algorithms. We perform Monte-Carlo Simulations for generating plots for Section \ref{ppn} and \ref{ag}, where each observed value is averaged over 20 separate simulations.

We first examine performance on real-world Bitcoin data~\cite{neudecker-atc16,neudecker-fc17} in Section~\ref{bn}. In Section~\ref{ppn}, we present a comparative study on three real-world peer-to-peer networks, namely - BitTorrent, Skype and Bitcoin~\cite{Stutzbach:2006:UCP:1177080.1177105,guha2005experimental,neudecker-atc16,neudecker-fc17} . Finally, we study the effect of joins and departures on the two algorithms presented in  Section~\ref{ag}, where we simulate session lengths governed by Weibull Distribution.


\begin{figure*}[t!]
\captionsetup{justification=justified,margin=1.5cm}
\begin{subfigure}{.5\textwidth}
  \centering
  \includegraphics[trim = {0cm 1cm 0cm 1cm },width= 0.999\linewidth]{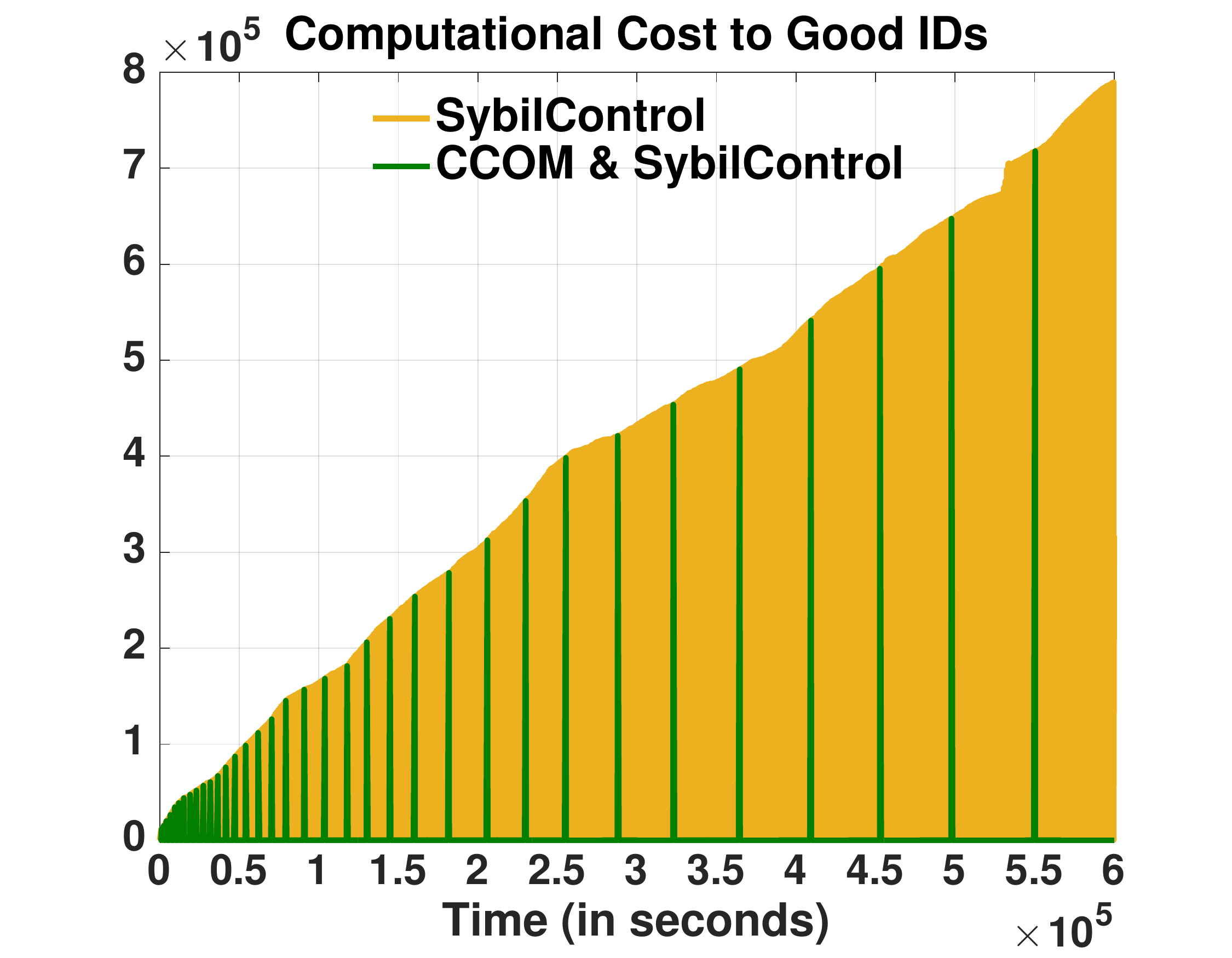}
  \caption{}
  \label{fig:1}
\end{subfigure}%
\begin{subfigure}{.5\textwidth}
  \centering
  \includegraphics[trim = {0cm 1cm 0cm 1cm },width= 0.999\linewidth]{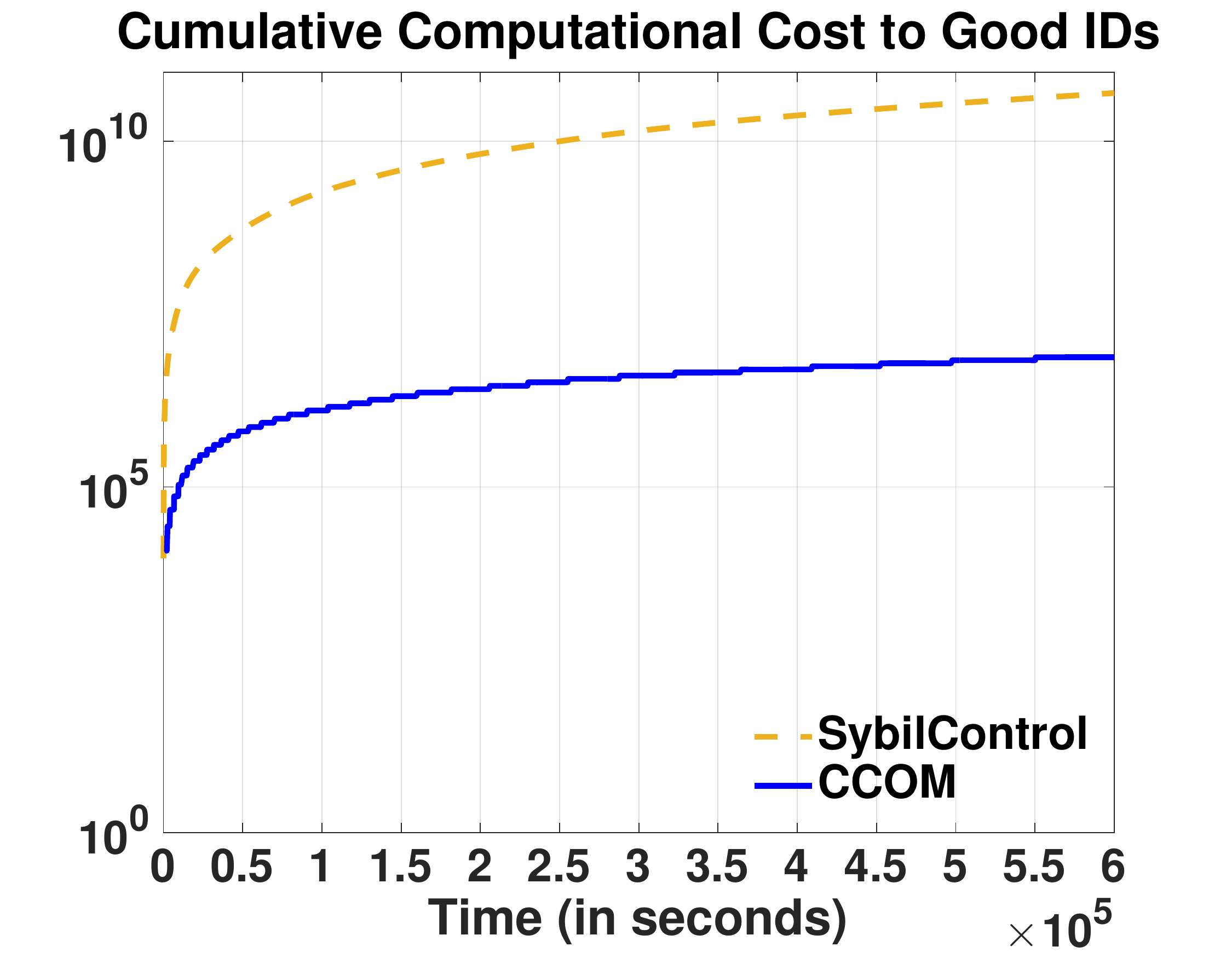}
  \caption{}
 \label{fig:2}
\end{subfigure}\vspace{-8pt}
\caption{Computational cost and cumulative computational for good IDs in absence of an attack using the Bitcoin Dataset.}
\label{fig:12}
\end{figure*}


\subsection{The Bitcoin Network}{\label{bn}}

We simulate~\AlgA~and~\AlgB~on a real-world dataset for the Bitcoin Network~\cite{7140490}. This data consists of roughly $1$ week of join/departure event timestamps.\footnote{The method of collection uses ongoing crawls from multiple seed machines in the Bitcoin network and captures IP addresses of reachable participants at $5$-minute intervals. Data obtained by personal correspondence with Till Neudecker (coauthor on~\cite{7140490}).}

The computational cost is examined under two scenarios: (1) in the absence of bad IDs (i.e., no attack) and (2) when bad IDs are present (i.e., an attack).  For our experiments, we assume that $\alpha = 1/10$.

For scenario (1), we assume all events in the dataset are the result of good IDs joining/departing the system. In Figure~\ref{fig:12}(a), the dense green area signifies the high frequency of computational cost paid by good IDs in~\AlgA~in comparison to the increasingly-spaced blue plot for \AlgB. Figure~\ref{fig:12}(b) shows the cumulative computational cost to the good IDs for~\AlgA~and \AlgB~in this scenario. {\bf Importantly, the cumulative cost to the good IDs is less by roughly 4 orders of magnitude} after $13$ hours of simulation time, and this gap continues to widen with time; note the logarithmic $y$-axis. This result demonstrates that our resource-competitive solution can be efficient in the absence of attack. 

For scenario (2), the same join and departure events occur as in scenario (1), but additionally an aggressive attack is orchestrated using the following adversarial strategy. From time $t/3$ seconds to $2t/3$ seconds, where $t = 604,970$, every $5$ seconds, the adversary adds a number of bad IDs that is a $1/3$ fraction of the current system size.  We note that the plot is over roughly $7$ days. Given this attack, a computational test is triggered in both algorithms every $5$ seconds, which aligns with the periodic testing in \AlgA.

In Figure~\ref{fig:34}(a), the dense green region represents the computational cost paid by the good IDs in \AlgA, the dense blue region represents that paid by good IDs in \AlgB and the red region by the adversary. Due to the setup of the stress test, the computational cost paid by good IDs in both the algorithms from $\frac{t}{3}$ to $\frac{2t}{3}$, is equivalent, thus the blue region overlaps the green.
Figure~\ref{fig:34}(b) depicts the cumulative cost to the good IDs for \AlgA, \AlgB, and the adversary, and we make two observations. First, the cost of \AlgB~is less than that of \AlgA. Second,  \AlgB's cost is indeed a function of the cost paid by the adversary; in contrast, without a resource-competitive guarantee, \AlgA's cost grows at a significantly faster rate (again, note the logarithmic $y$-axis).

\begin{figure*}[t!]
\captionsetup{justification=justified,margin=1.5cm}
\begin{subfigure}{.5\textwidth}
  \centering
  \captionsetup{justification=centering,margin=1.5cm}
  \includegraphics[trim = {0cm 1cm 0cm 1cm},width= 0.999\linewidth]{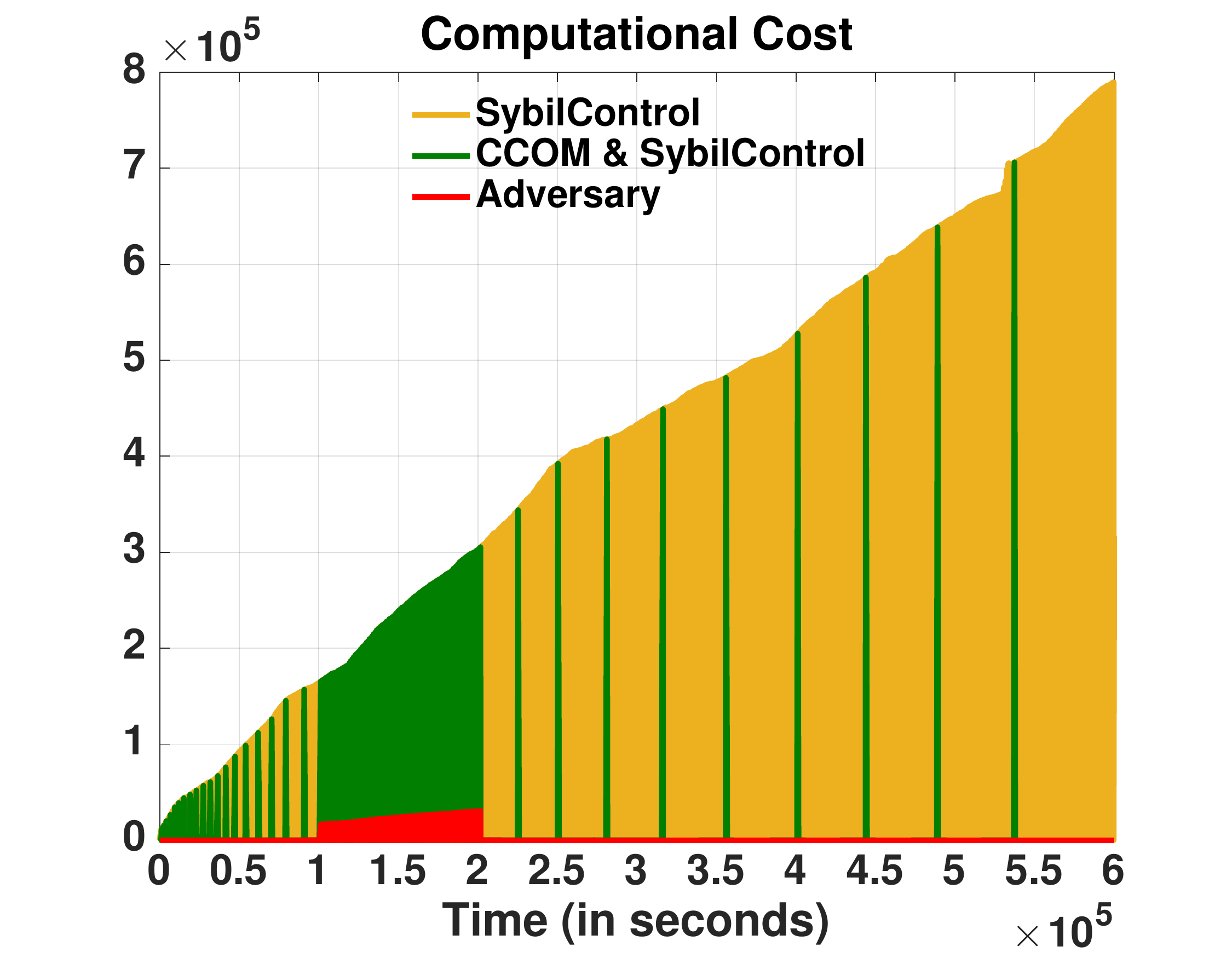}
  \caption{}
  \label{fig:3}
\end{subfigure}%
\begin{subfigure}{.5\textwidth}
  \centering
  \includegraphics[trim = {0cm 1cm 0cm 1cm },width= 0.999\linewidth]{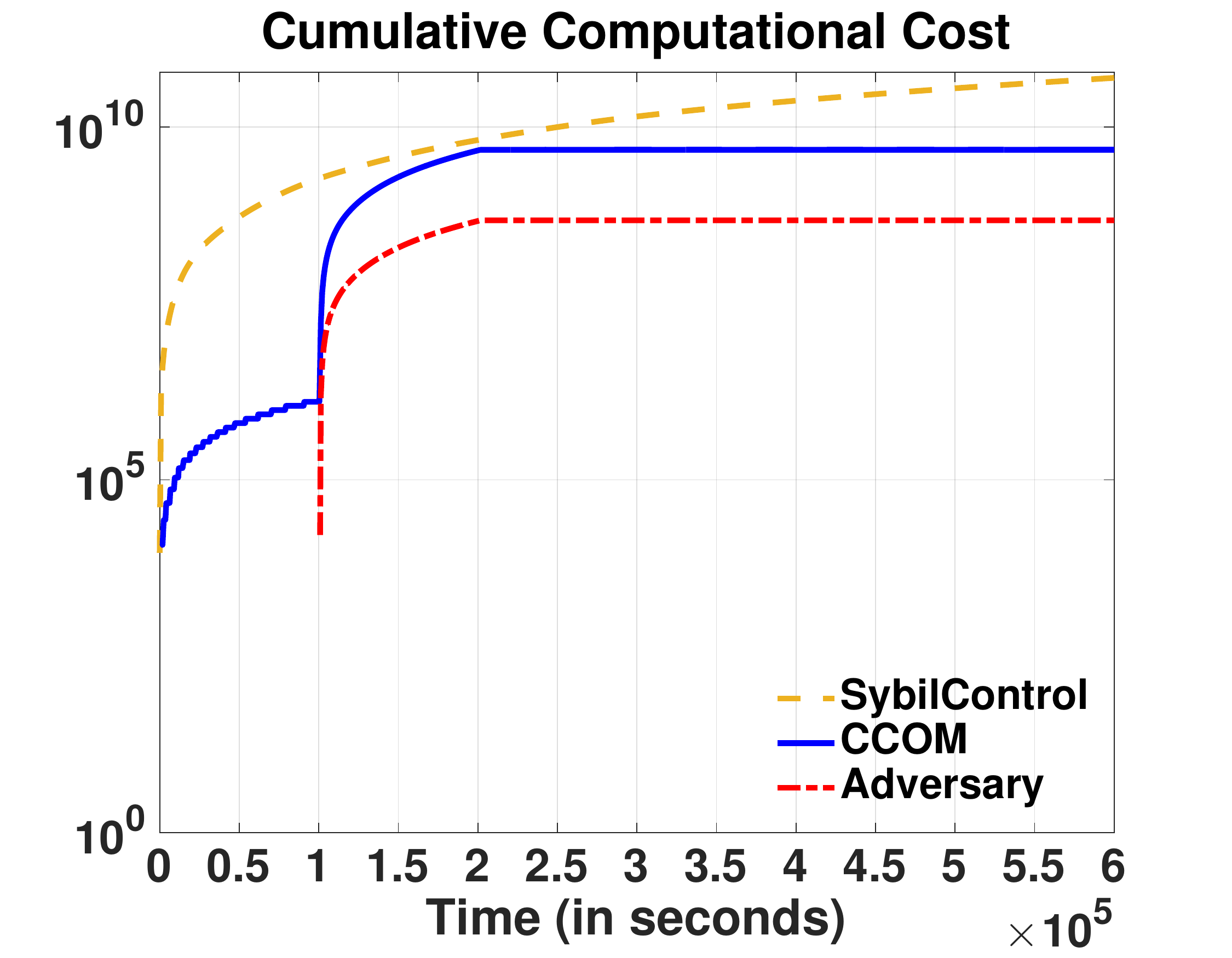}
  \caption{}
 \label{fig:4b}
\end{subfigure}\vspace{-8pt}
\caption{Computational cost and cumulative computational cost versus time under a large-scale attack using the  Bitcoin Dataset.}
\label{fig:34}
\end{figure*}

\setlength{\intextsep}{0pt}
\setlength{\columnsep}{0pt}
\begin{wrapfigure}[14]{r}[10pt]{.5\textwidth}
\centering
\vspace{-1cm}\includegraphics[width=0.9\textwidth]{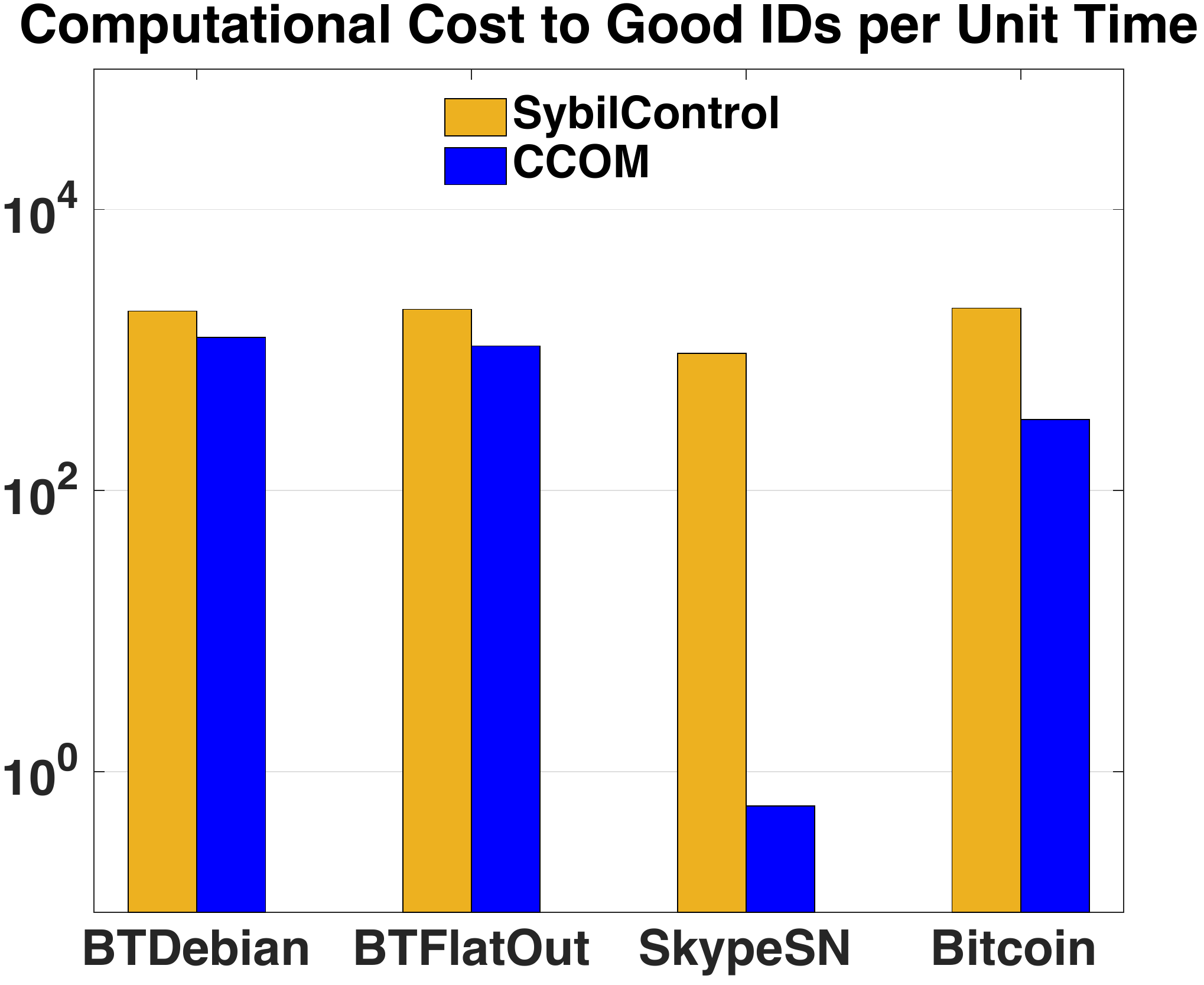}
\vspace{-0.2cm}\captionsetup{width=0.75\textwidth}
\caption{Computational Cost per unit Time for Peer-to-Peer Networks.}\label{fig:SCON-loses}
\end{wrapfigure}

\subsection{Application to BitTorrent and Skype}\label{ppn} 

To illustrate the savings obtained by our approach in additional network settings when there is no attack, we compare performance on three different networks, namely  BitTorrent, Skype and Bitcoin.  For Skype, we consider the network of supernodes; we refer to this by SkypeSN. For BitTorrent, we performed experiments on two BitTorrent overlays associated with (i) Debian ISO images, and (ii) a demo of the game FlatOut; we refer to these as BTDebian and BTFlatOut, respectively. 

For each network, respective session times are dictated by a Weibull distribution with parameters set according to the empirical findings in~\cite{Stutzbach:2006:UCP:1177080.1177105,guha2005experimental,7140490}. The shape ($k$) and scale ($\lambda$) parameter values of Weibull Distribution for session time of $k$ = 0.38 and $\lambda$ = 42.2 for Debian, and $k$ = 0.59 and $\lambda$ = 41.9 for FlatOut. The Skype supernodes had also have a Weibull Distribution for session time with median session time of 5.5 hours, and shape parameter of 0.64. We generate the session time for 10,000 good IDs from these parameter values and sample for the bitcoin network from the real-world data obtained from~\cite{neudecker-atc16,neudecker-fc17}.

The computational cost to the good IDs per unit time is plotted for \AlgA~and \AlgB in Figure~\ref{fig:SCON-loses}. Note again the logarithmic $y$-axis. We observe that \AlgB~outperforms \AlgA~in all four cases in terms of computational costs of the network. \AlgB~outperforms \AlgA~by 34.5\% in BitTorrent Debian, by 45.6\% in BitTorrent FlatOut, by 99.9\% in Skype Supernodes and 83.9\% in Bitcoin. We use the following definition of performance: $100
\cdot \left(1-\frac{\AlgB}{\AlgA}\right)$.

\begin{figure*}[t!]
\captionsetup{justification=justified,margin=1.5cm}
\begin{subfigure}{.5\textwidth}
  \centering
  \includegraphics[trim = {0cm 2cm 0cm 2cm }, width= 0.999\linewidth]{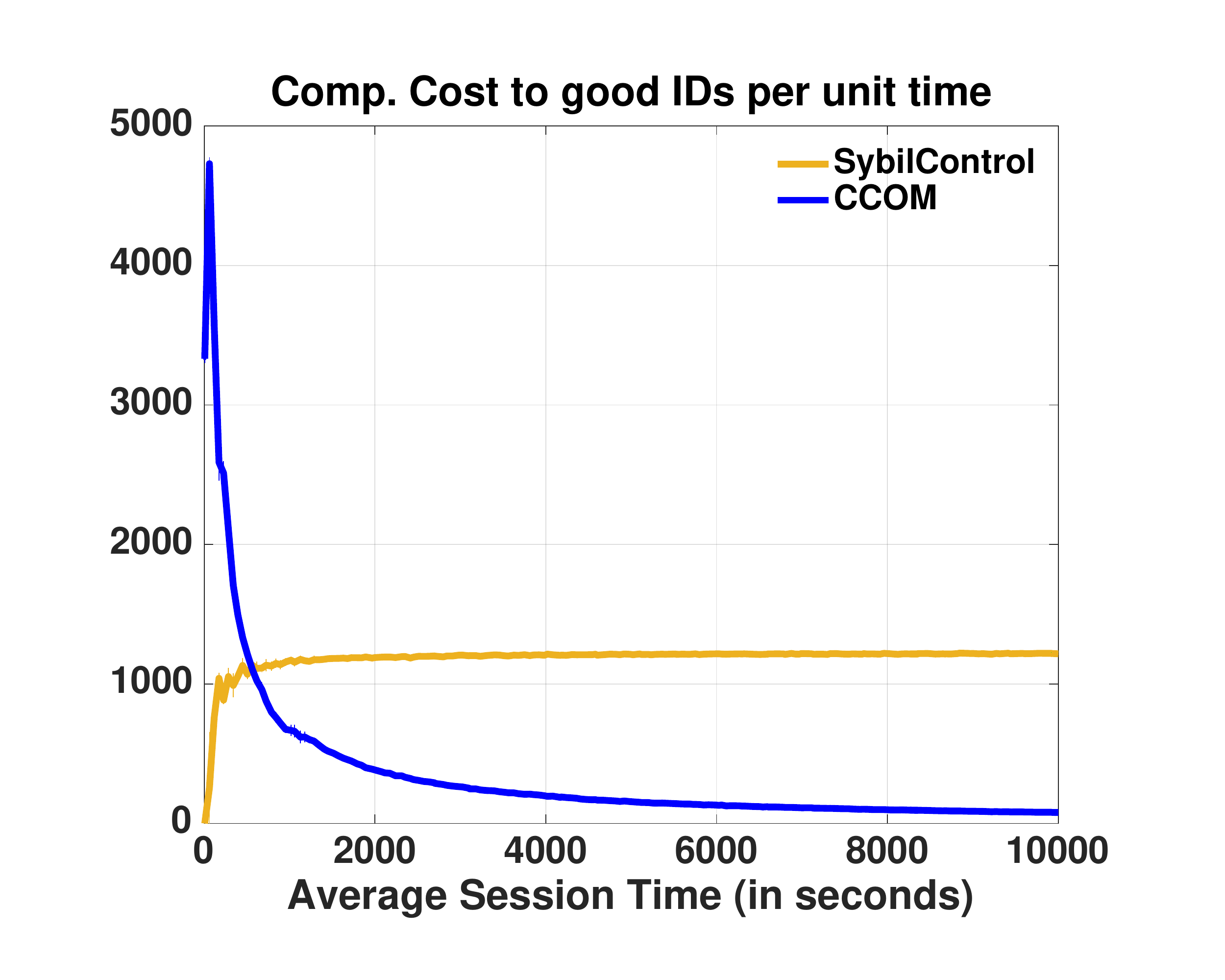}
  \caption{}
  \label{fig:4}
\end{subfigure}%
\begin{subfigure}{.5\textwidth}
  \centering
  \includegraphics[trim = {0cm 2cm 0cm 2cm },width= 0.999\linewidth]{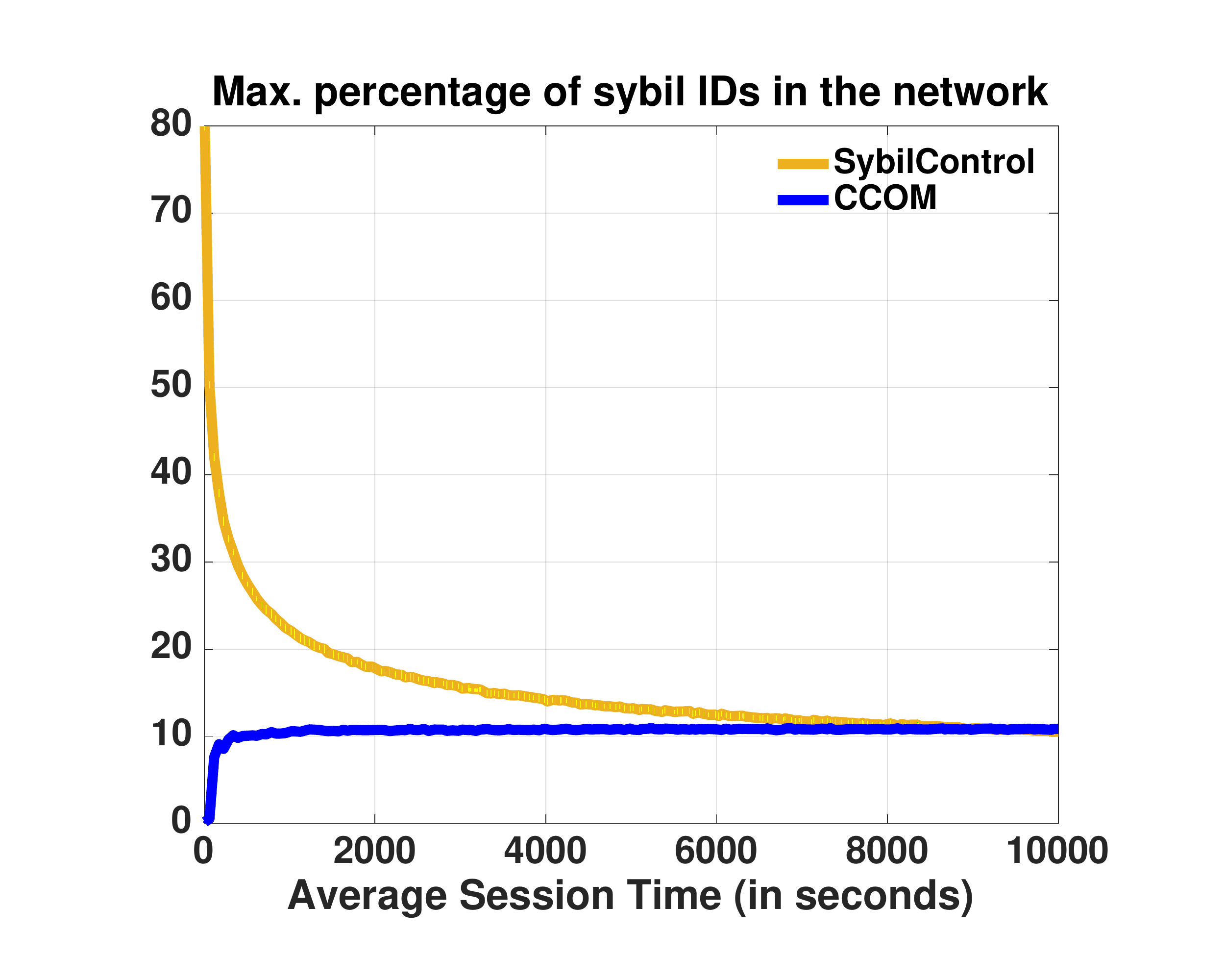}
  \caption{}
 \label{fig:5}
\end{subfigure}
\caption{Computational Cost per unit time and Maximum Percentage of bad IDs in system versus Average Session length.}
\label{fig:45}
\end{figure*}

\subsection{Effect of Joins and Departures}{\label{ag}}
 
 
We examine the effect of churn on the performance of \AlgB~versus \AlgA. Intuitively, \AlgB~should perform best relative to \AlgA~when the churn is low since purges will be infrequent (whereas \AlgA~is continually performing tests every $5$ seconds). In contrast, under significant churn, the performance gap between these two algorithms may narrow. Furthermore, when there is significant churn, \AlgA~may be unable to limit the number of bad IDs given that a $5$-second interval between tests is insufficient. In this section, we present the results of our experiments in order to better understand these aspects of performance.
 
Empirical studies of session lengths in many overlay networks indicate that the Weibull Distribution~\cite{weibull1951wide} is a better fit  than the often-used Poisson Distribution~\cite{Stutzbach:2006:UCP:1177080.1177105}. Hence, in our experiments, we used the Weibull distribution to generate the session lengths of IDs. 

The average session length ranged from as low as $0.1$ seconds to $10,000$ seconds.  The system was initialized with $10,000$ good IDs and the simulations were carried out until $15,000$ new IDs joined the network. Each new ID could be good or bad with probability $0.5$. For each value of average session length, we took the mean of our observations over $30$ runs. 

Our results are depicted in Figure~\ref{fig:45}. With the increase in the average session length (on the order of hours), \AlgB~is able to give similar guarantees as \AlgA~in terms of the maximum percentage of bad IDs in the system at much lower computation cost to the good IDs. On the other hand, when session length are small (order of minutes), \AlgB~is able to maintain a constant percentage of bad IDs in the system whereas \AlgA~is not able to give any such guarantees; indeed, the percentage of bad IDs can become unbounded.


\section{Applications}\label{sec:applications}

\subsection{Scalable Byzantine Consensus}

The problem of  Byzantine Consensus (BC) was introduced by Lamport, Shostak and Pease\cite{lamport1982byzantine}; this problem is also referred to as  Byzantine Agreement. There are $n$ IDs, of which some hidden subset are \emph{bad}.  These bad IDs may deviate from a prescribed algorithm in an arbitrary way. The remaining IDs are \emph{good}, and will follow our algorithm faithfully.   Each good ID has an initial input bit. The goal is for (1) all good IDs to decide on the same bit; and (2) this bit to equal the input bit of at least one good ID.

Byzantine consensus enables the creation of a reliable system from unreliable components.  Therefore, it is not surprising that BC is used in areas as diverse as: maintaining blockchains~\cite{BonneauMCNKF15,eyal2016bitcoin,cryptoeprint:2015:521,Micali16}; trustworthy computing~\cite{castro1998practical, castro2002practical,SINTRA,cachin:secure,kotla2007zyzzyva,clement-making, 1529992}; P2P networks~\cite{oceanweb,adya:farsite}; and databases~\cite{GPD,preguica2008byzantium, zhao2007byzantine} \smallskip

A major shortcoming of current algorithms for BC is that they do not scale.  For example, Rhea \textit{et al}. write: \emph{``Unfortunately, Byzantine agreement requires a number of messages quadratic in the number of participants, so it is infeasible for use in synchronizing a large number of replicas''}~\cite{rhea2003pond} (see also~\cite{cowling:hq},\cite{Cheng2009219}). This quadratic message cost hinders deployment in modern systems where the number of participants can be large.  
 
 Recent results reduce the total number of messages to $\tilde{O}(n^{3/2})$~\cite{PODC10, KSJACM}.  However, these algorithms (1) have high cost in practice; (2) are complicated to implement; and (3) require non-constructive combinatorial objects such as expander graphs. \medskip

\noindent{\bf How Our Result Applies.}  Our result allows us to reduce the communication cost for BC. The bad IDs are again incarnated as a single adversary with equivalent computational power. Each good participant has a single ID, while the adversary is not constrained in its creation of IDs. 

The members of the \qcomm~execute any BC algorithm that requires a quadratic number of messages; this implies a message cost of $O(\log^2 n)$ generated by the \qcomm. The \qcomm~then diffuses the signed consensus value to all other IDs. On receiving these agreed upon values from their \qcomm~members, the non-\qcomm~members take the majority of these verified consensus values. Since \qcomm~has a majority of good members, this results in all IDs holding the correct consensus value.

In this way, we are able to solve traditional BC with $\tilde{O}(n)$ number of messages in total. Additionally, we can solve a dynamic version of BC, where IDs are joining and leaving, and can do so with computational cost to the good IDs that is commensurate with the cost incurred by the adversary.

\subsection{\E: Committee Election}{\label{Ealgb}}


\E~is a secure  protocol~\cite{Luu:2016}, which aims to achieve agreement on a set of transactions in a blockchain i.e., the state of the blockchain.  Informally, the core idea is to partition the system into smaller fragments called committees, where each committee executes a BC protocol to agree upon its set of transactions. Then, the committee that was formed last --- referred to as the \defn{final committee} --- computes the final digest of all transactions in the system; these transactions having been received from other committees. This final digest is  broadcast to all other participants in the system. \medskip 

\noindent{\bf How Our Result Applies.} In order to reduce message cost,  \E~makes use of a special committee referred to as the \defn{directory committee (DC)}, which coordinates the formation of all other committees. We propose the election of the DC using \AlgB with the last portion of Line 1 is amended: there is no need for the committee to verify solutions or maintain $\Current$ (recall Section~\ref{sec:sys-mem}). This ensures the \cgoal and our cost results; that is, we can guarantee that (1) the committee contains a majority of good IDs for a polynomial number of join and leave events; and (2) bandwidth and computational costs grow commensurately with the costs of the adversary. 


\section{Conclusion}
We have described algorithms to efficiently use \POW~computational puzzles to reduce the fraction of bad IDs in open systems.  Unlike previous work, our algorithms require the good IDs to expend computational resources that grow only linearly with the computational resources expended by the adversary.  In particular, assume the adversary incurs computational cost $T_C$ and sends $T_B$ messages, and that $g_{new}$ good IDs enter the system.  Then our algorithm requires $O(T_C + g_{new})$ computational cost and sends $O(T_B + g_{new})$ messages.

Many open problems remain.  An immediate question is whether we can improve the asymptotic costs for our algorithm, or prove a lower bound of $\Omega(T + \gnew)$.  Our intuition is that the former is possible, although proving this appears to be non-trivial. For example, in order to reduce the frequency with which Byzantine Consensus is solved,  a plausible approach is to adopt the concept of a core set from~\cite{KKKSS}. Informally, a core set captures the notion that good IDs can have different views of {\it good} committee members, but that a sufficiently large overlap between these views exists such that correctness is still guaranteed.

One of particular interest is: Can we adapt our technique to secure multi-party computation? The problem of \defn{secure multi-party computation (MPC)} involves designing an algorithm for the purpose of computing the value of an $n$-ary function $f$ over private inputs from $n$  IDs $x_1, x_2,...,x_n$, such that the IDs learn the value of $f(x_1,x_2,...,x_n)$, but learn nothing more about the inputs than what can be inferred from this output of $f$.  The problem is generally complicated by the assumption that an adversary controls a hidden subset of the IDs.  In recent years, a number of attempts have been made to solve this problem for very large $n$ manner~\cite{applebaum2010secrecy, beerliova2006efficient, bogetoft2009secure, damgard2006scalable, damgaard2008scalable, dani2014quorums, goldreich1998secure}.  We believe that our technique for forming committees in a dynamic network could be helpful to solve a dynamic version of this problem, while ensuring that the resource costs to the good IDs is commensurate with the resource costs to an adversary.

As another avenue of future work, there are scenarios where it may be unreasonable to assume that ``good'' IDs always blindly follow our algorithm.  Instead, these IDs may be rational but selfish, in that they seek to optimize some known utility function.  The adversary still behaves in a worst-case manner, capturing the fact that the utility function of the adversary may be completely unknown.  This approach is similar to BAR (Byzantine, Altruistic, Rational) games~\cite{clement:theory} in distributed computing, see also~\cite{abraham:distributed, Aiyer:2005:BFT,Li:2006:BG,Vilaca:2012:ACN,Vilaca:NBT,Wong:2010:MBA}.  Extending our result to accommodate this model is of interest.

Finally, there is a substantial body of literature on attack-resistant overlays; for example~\cite{fiat:making,awerbuch:towards,awerbuch:random,awerbuch:towards2,young:practical,saia:reducing,young:towards,guerraoui:highly,naor:novel,naor_wieder:a_simple,sen:commensal,saad:self-healing,saad:self-healing2,awerbuch_scheideler:group,AnceaumeLRB08,HK}. These results critically depend on the fraction of bad IDs  always being upper-bounded by a constant less than $1/2$; however, there are only a handful of results that propose a method for guaranteeing this bound; see Section~\ref{sec:related-work}. A natural idea is to examine whether \AlgB can be applied to this setting in order to guarantee this bound (via the \sgoal).


\end{document}